\documentclass[pdftex,pagesize=auto,paper=letter,11pt,abstract=on,bibliography=oldstyle]{scrartcl}

\usepackage[utf8]{inputenc}
\usepackage[T1]{fontenc}
\usepackage[english]{babel}
\usepackage{amsthm}
\usepackage{amsmath}
\usepackage{amssymb}
\usepackage{hyperref}
\usepackage{graphicx}
\usepackage{xspace}
\usepackage{xcolor}
\usepackage{authblk}
\usepackage{multirow}
\usepackage{booktabs}
\usepackage{paralist}
\usepackage{caption}
\usepackage{subcaption}
\usepackage{tikz}

\usepackage[linesnumbered,noend,ruled]{algorithm2e}
\DontPrintSemicolon

\newtheorem{lemma}{Lemma}

\DeclareMathOperator{\weight}{len}
\DeclareMathOperator{\cons}{cons}

\newcommand{\lef}{\ell}
\newcommand{\rig}{r}

\newcommand{\eg}{e.\,g.\xspace}
\newcommand{\ie}{i.\,e.\xspace}
\newcommand{\wrt}{wrt.\xspace}

\newcommand{\graph}{\ensuremath{G}}
\newcommand{\vertices}{\ensuremath{V}}
\newcommand{\edges}{\ensuremath{E}}
\newcommand{\vertex}{\ensuremath{v}}
\newcommand{\vertexa}{\ensuremath{u}}
\newcommand{\vertexb}{\ensuremath{v}}

\newcommand{\edge}{\ensuremath{e}}

\newcommand{\reals}{\ensuremath{\mathbb{R}}}
\newcommand{\posreals}{\ensuremath{\reals_{\ge 0}}}

\newcommand{\apath}{\ensuremath{\pi}\xspace}
\newcommand{\source}{\ensuremath{s}}

\newcommand{\planar}{\ensuremath{p}}

\newcommand{\distancelabel}{\ensuremath{d}}
\newcommand{\consumptionlabel}{\ensuremath{c}}

\newcommand{\bigO}{\ensuremath{\mathcal{O}}}
\usepackage{complexity}
\newcommand{\opt}{\mathrm{OPT}}

\newcommand{\range}{\ensuremath{r}}

\newcommand{\algoExtractReachableComponent}{\texttt{RP-RC}\xspace}
\newcommand{\algoTriangularSeparators}{\texttt{RP-TS}\xspace}
\newcommand{\algoConnectComponents}{\texttt{RP-CU}\xspace}
\newcommand{\algoSelfIntersections}{\texttt{RP-SI}\xspace}

\newcommand{\email}[1]{\texttt{#1}}
\title{Scalable Isocontour Visualization in\\Road Networks via Minimum-Link Paths\thanks{Partially supported by the EU FP7 under grant agreement no.\ 609026 (project MOVESMART)}}
\author[1]{Moritz Baum}
\author[1,2]{Thomas Bläsius}
\author[1]{Andreas Gemsa}
\author[1]{Ignaz Rutter}
\author[1]{Franziska Wegner}
\affil[1]{Karlsruhe Institute of Technology, Germany\\\email{firstname.lastname@kit.edu}}
\affil[2]{Hasso Plattner Institute, Germany\\\email{thomas.blaesius@hpi.de}}
\date{}

\begin{document}

\maketitle

\begin{abstract}
Isocontours in road networks represent the area that is reachable from a source within a given resource limit. We study the problem of computing accurate isocontours in realistic, large-scale networks.
We propose polygons with minimum number of segments that separate reachable and unreachable components of the network.
Since the resulting problem is not known to be solvable in polynomial time, we introduce several heuristics that are simple enough to be implemented in practice. A key ingredient is a new practical linear-time algorithm for minimum-link paths in simple polygons.
Experiments in a challenging realistic setting show excellent performance of our algorithms in practice, answering queries in a few milliseconds on average even for long ranges.
%
%
%
%
\end{abstract}

\section{Introduction}
\label{sec:introduction}


How far can I drive my battery electric vehicle, given my
position and the current state of charge?  -- This
question expresses range anxiety (the fear of getting
stranded) caused by limited battery capacities and sparse
charging infrastructure.  An answer in the form of a map that
visualizes the reachable region helps to find charging stations in
range and to overcome range anxiety.  This reachable region is bounded
by curves that represent points of constant energy consumption; such
curves are usually called \emph{isocontours} (or \emph{isolines}).
Isocontours are typically considered in the context of functions
$f\colon \mathbb R^2 \to \mathbb R$, e.g., if~$f$ describes the
altitude in a landscape, then the terrain can be visualized by showing
several isocontours (each representing points of constant altitude).
In our setting, $f$ would describe the energy necessary to reach a
certain point in the plane.  However, $f$ is actually defined only for
a discrete set of points, namely for the vertices of the road network.
Thus, we have to fill the gaps by deciding how the isocontour should
pass through the regions between the roads.  The fact that the quality
of the resulting visualization heavily depends on these decisions
makes computing isocontours in road networks an interesting
algorithmic problem.

Formally, we consider the road network to be given as a directed graph
$\graph = (\vertices,\edges)$,
along with vertex positions in the plane and two weight
functions \mbox{$\weight \colon E \to \mathbb{R}^+$} and
$\cons \colon E \to \mathbb{R}$ representing length and resource
consumption, respectively.  For a source
vertex~$\source \in \vertices$ and a range~$\range \in \posreals$, a
vertex $v$ belongs to the \emph{reachable subgraph} if the shortest
path from $s$ to $v$ has a total resource consumption of at most~$\range$,
i.e., shortest paths are computed according to the length, while
reachability is determined by the consumption.
Coming back to our initial question concerning the electric vehicle,
the source vertex is the initial position, the range is the current
state of charge, the length corresponds to travel time, and energy is
the resource consumed on edges.  We allow negative resource
consumption to take account of recuperation, which enables electric
vehicles to recharge their battery when braking.  Note that our
setting is sufficiently general to allow for other applications.  By
setting the length as well as the resource consumption of edges to the
travel time ($\weight \equiv \cons$), one obtains the special case of
\emph{isochrones}.  There is a wide range of applications for
isochrones, including reachability analyses~\cite{Bau08,Gam11,Gam12},
geomarketing~\cite{Efe13}, and various online
applications~\cite{Inn13}.
Known approaches focus on isochrones of small or medium
range. But isochrones can be useful in more challenging scenarios, for example, to visualize the area reachable by a truck driver within a day of work. This motivates our work on
fast isocontour visualization.

Our algorithms for computing isocontours in road networks
are guided by three major objectives.  The isocontours must be exact
in the sense that they correctly separate the reachable subgraph from
the remaining \emph{unreachable subgraph}; the isocontours should be
polygons of low \emph{complexity} (\ie, consist of few segments, enabling efficient rendering and a clear, uncluttered visualization); and the algorithms should be sufficiently
fast in practice for interactive applications, even on large inputs of continental scale.

\begin{figure}[t]
\centering
\begin{subfigure}{.48\textwidth}
\includegraphics[width=\textwidth]{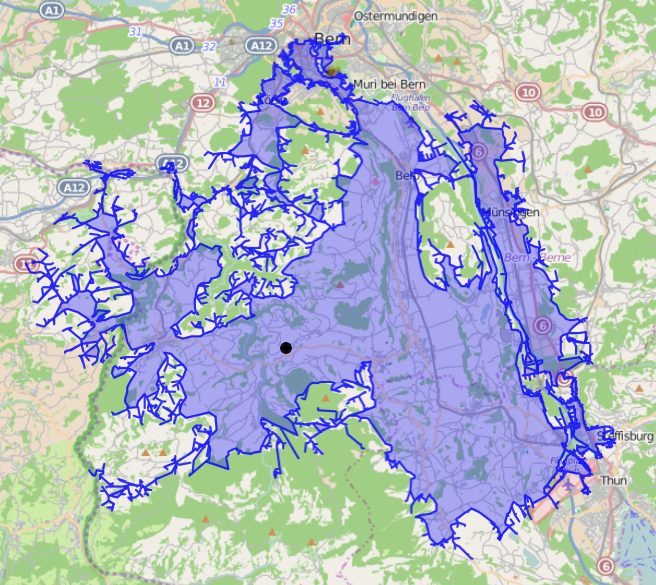} 
\caption{}
\end{subfigure}\hfill
\begin{subfigure}{.48\textwidth}
\includegraphics[width=\textwidth]{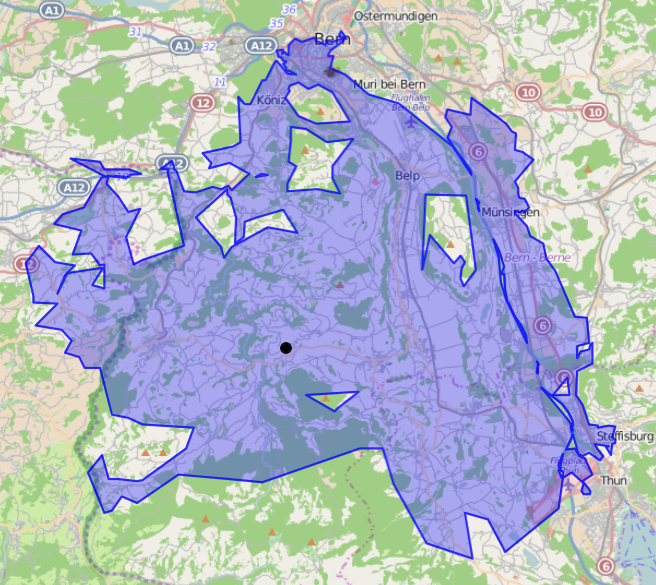} 
\caption{}
\end{subfigure}
\caption{Real-world example of isocontours in a mountainous area (near
  Bern, Switzerland). Both figures visualize range of an electric
  vehicle positioned at the black disk with a state of charge
  of~2\,kWh. Note that the polygons representing the isocontour
  contain holes, due to unreachable high-ground areas. (a) The isocontour computed by the approach in Section~\ref{sec:heuristics:reachable_component_extraction}, which resembles previous works~\cite{Mar10}. (b) The result of one of
  our new approaches, presented in Section~\ref{sec:heuristics:connected_components}.}
\label{fig:case_study_polygon_with_holes}
\end{figure}

Figure~\ref{fig:case_study_polygon_with_holes} shows an example of
isocontours visualizing the range of an electric vehicle. The left
figure depicts a polygon that closely resembles the output of
isocontour algorithms considered state-of-the-art in recent
works~\cite{Efe13,Mar10}. Unfortunately, the number of segments
becomes quite large even in this medium-range example (more than 10\,000 segments). The figure to the right shows the result of our
approach presented in
Section~\ref{sec:heuristics:connected_components}, which contains the
same reachable subgraph while using 416 segments in total.

\paragraph{Related Work.}

Several existing algorithms consider the problem of computing the subnetwork that is reachable within a given timespan.
The MINE algorithm~\cite{Gam11} is a search to compute isochrones in transportation networks based on Dijkstra's well-known algorithm~\cite{d-ntpcg-59}. An improved variant, called MINEX~\cite{Gam12}, reduces space requirements. Both approaches work on spatial databases, prohibiting interactive applications (having running times in the order of minutes for large ranges).
To enable much faster shortest-path queries in practice, \emph{speedup techniques}~\cite{Bas14} seperate the workflow into an offline preprocessing phase and an online query phase. Recently, some speedup techniques were extended to the isochrone scenario~\cite{bbdw-fcirn-15,ep-grasp-14}, enabling query times in the order of milliseconds.
Yet, these approaches only deal with the computation of the reachable subgraph, rather than visualization of isocontours.
%

Regarding isocontour visualization, efficient approaches exist for shape characterization of point sets, such as $\alpha$-shapes~\cite{Ede82} or $\chi$-shapes~\cite{Duc08}. However, we are interested in separating subgraphs rather than point sets.
Marciuska and Gamper~\cite{Mar10} present two approaches to visualize isochrones. The first  transforms a given reachable network into an isochrone area by simply drawing a buffer around all edges in range. The second one creates a polygon without holes, induced by the edges on the boundary of the embedded reachable subgraph.
Both approaches were implemented on top of databases, providing running times that are too slow for many applications (several seconds for small and medium ranges).
Gandhi~et~al.~\cite{Gan07} introduce an algorithm to approximate isocontours in sensor networks with provable error bounds. They preserve the topological structure of the given family of contours, forbidding intersections.
Finally, there are different works presenting applications which make use of isocontours in the context of urban planning~\cite{Bau08}, geomarketing with integrated traffic information~\cite{Efe13}, and range visualization for electric vehicles~\cite{Gru14}.

\paragraph{Contribution and Outline.}

We propose algorithms for efficiently computing
polygons representing isocontours in road networks. All approaches
compute isocontours that are exact, \ie, they contain exactly
the subgraph reachable within the given resource
limit, while having low descriptive complexity.
%
%
Efficient performance
of our techniques is both proven in theory and demonstrated in
practice on realistic instances.
%

In Section~\ref{sec:problem_statement}, we formalize the notion of reachable and unreachable subgraphs. Moreover, we state the precise problem and outline our algorithmic approach to solve it.
Section~\ref{sec:range_query} attacks the first resulting subproblem of computing \emph{border regions}, that is, polygons that represent the boundaries of the reachable and unreachable subgraph.
An isocontour must separate these boundaries.
In Section~\ref{sec:min_link_path}, we consider the special case of separating two hole-free polygons with a polygon with minimum number of segments. While this problem can be solved in $O(n \log n)$ time~\cite{Wan91},
%
%
we propose a
simpler algorithm that uses at most two additional segments, runs in linear time, and requires a single run of a minimum-link path algorithm. We also propose a minimum-link path algorithm that is simpler than a previous approach~\cite{Sur86}.
Section~\ref{sec:heuristics} extends these results to the general case, where border regions may consist of more than two components. Since the complexity of the resulting problem is unknown, we focus on efficient heuristic approaches that work well in practice, but do not give guarantees on the complexity of the resulting range polygons.
Section~\ref{sec:experiments} contains our extensive experimental evaluation using a large, realistic input instance. It demonstrates that all approaches are fast enough even for use in interactive applications.
We close with final remarks in Section~\ref{sec:conclusion}.

\section{Problem Statement and General Approach}
\label{sec:problem_statement}



Let $\graph=(\vertices,\edges)$ be a road network, which we consider as a geometric
network where vertices have a fixed position in the plane and edges
are represented by straight-line segments between their endpoints.  A
source $\source \in \vertices$ and a maximum range $\range$ together
partition the network into two parts, one that is within range $\range \in \posreals$
from~$v$, and the part that is not.  An isocontour separates these two
parts.  We are interested in visualizing such isocontours efficiently.  In
the following, we give a precise definition of the (un)reachable parts
of the network and formally define \emph{range polygons}, which we use
to represent isocontours.

\paragraph{Range Polygons.}

A path $\apath$ in $\graph$ starting at $\source$ is \emph{passable} if the consumption
of $\apath$, \ie, the sum of its edge consumption values, is at most~$\range$.  A vertex $\vertex$ is \emph{reachable}
(with respect to the maximum range $\range$) if the shortest
$\source$--$\vertex$-path is passable.  A vertex that is not reachable is
\emph{unreachable}.
For edges the situation is more complicated.  We partition the edges
into four types, namely unreachable edges, boundary edges, accessible
edges, and passable edges. Figure~\ref{fig:reachability_categories}
shows an example of the different edge types. If both endpoints of an edge $(u,v)$ are
unreachable, then also the edge $(u,v)$ is \emph{unreachable}.  If exactly
one endpoint is reachable, then $(u,v)$ is not part of the reachable
network, and we call it \emph{boundary edge}.  However, the fact that
both~$u$ and~$v$ are reachable does not necessarily imply that
$(u,v)$ is part of the reachable network.  Let $\apath_u$ and $\apath_v$
denote the shortest paths from $s$ to $u$ and $v$, respectively.  If
the resource consumptions of $\apath_u$ and $\apath_v$ do not allow the
traversal of the edge $(u,v)$ in either direction, i.e.,
$\cons(\apath_u) + \cons((u,v)) > \range$ and $\cons(\apath_v) + \cons((v,u)) > \range$,
then we do not consider $(u,v)$ as reachable.  Since we can reach
both endpoints of $(u,v)$, we call it \emph{accessible}.  Otherwise, the edge can be traversed in at least one direction,
so it is \emph{passable}.

Let $\vertices_r$ be the reachable vertices and let
$\vertices_u = \vertices \setminus \vertices_r$ be the unreachable
vertices.  Similarly, let $\edges_u, \edges_b, \edges_a, \edges_r$
denote the set of unreachable edges, boundary edges, accessible edges,
and passable edges, respectively.
Note that for arbitrary pairs of edges~$(u,v)$ and $(v,u)$, both edges belong to the same set.
The reachable part of the network
is~$\graph_r = (\vertices_r,\edges_r)$, and the unreachable part
is~$\graph_u = (\vertices_u,\edges_u)$.  A \emph{range polygon} is a
plane (not necessarily simple) polygon $P$ separating $\graph_r$ and
$\graph_u$ in the sense that its interior contains~$\graph_r$ and has
empty intersection with~$\graph_u$.  Note that every range polygon $P$
intersects each boundary edge an odd number of times and each
accessible edge an even number of times.  In particular, an accessible
edge may be totally or partially contained in the interior of a range
polygon.

\begin{figure}[t]
  \centering
  \includegraphics{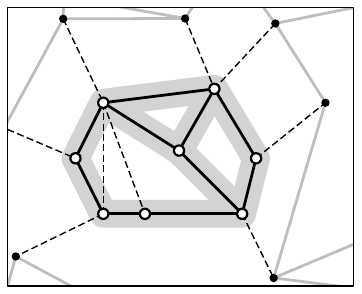}
  \caption{Graph with reachable (white) and unreachable (black) vertices. Black edges are passable. Dashed edges are boundary edges if they have an unreachable endpoint, and accessible if both endpoints are reachable. Gray edges are unreachable.}
  \label{fig:reachability_categories}
\end{figure}

If the input graph~$\graph$ is planar, one can construct a range
polygon by slightly shrinking the faces of the subgraph induced by all
reachable vertices, though this may produce many holes; see the shaded
area in Figure~\ref{fig:reachability_categories}.  However, if~$\graph$
is not planar, a range polygon may not even exist. If a passable edge
intersects an unreachable edge, the requirements of including the
passable and excluding the unreachable edge obviously contradict.  To
resolve this issue, we consider the planarization $\graph_\planar$ of
$\graph$, which is obtained from $\graph$ by considering each
intersection point $p$ as a \emph{dummy vertex} that subdivides all
edges of $\graph$ that contain~$p$.  We transfer the above partition
of $\graph$ into reachable and unreachable parts to $\graph_\planar$
as follows.  A dummy vertex is \emph{reachable} if and only if it
subdivides at least one passable edge of the original graph.  As
above, an edge of $\graph_\planar$ is unreachable if both endpoints
are unreachable, and it is a boundary edge if exactly one endpoint is
reachable.  If both endpoints are reachable, it is accessible
(passable) if and only if the edge in $G$ containing it is
accessible (passable).  Clearly, after the planarization, a range
polygon always exists. Figure~\ref{fig:crossing_edges} shows
different cases of crossing edges.  Note that this way of handling
crossings ensures that a range polygon for $\graph_\planar$ contains
the reachable vertices of $\graph$ and excludes the unreachable
vertices of~$\graph$.  However, unreachable edges of~$\graph$ may be
partially contained in the range polygon if they cross passable edges.
Finally, to avoid special cases, we add a bounding box of dummy
vertices and edges to~$\graph_\planar$, connecting each vertex in the bounding box to its respective closest vertex in~$\graph_\planar$ with an edge of infinite length. Thereby, we ensure that neither the
reachable nor the unreachable subgraph is empty, as the
reachable (unreachable) subgraph contains at least the source (the bounding box).

\begin{figure}[t]
  \centering
  \includegraphics{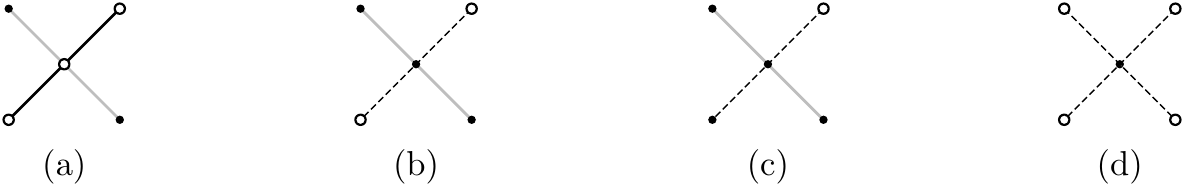}
  \caption{(a) Intersection of a passable and an unreachable edge, the dummy vertex is reachable. (b) Intersection of an accessible and an unreachable edge. The dummy vertex is unreachable, dashed edges become boundary edges. (c) Intersection of an unreachable edge and a boundary edge, creating an unreachable dummy edge and a new boundary edge replacing the original one. (d) An intersection of accessible edge creates an unreachable dummy vertex and four new boundary edges. }
  \label{fig:crossing_edges}
\end{figure}

\paragraph{General Approach.}

We seek to compute a range polygon with respect to the planarized
graph $G_p$ that has the minimum number of holes, and among these
we seek to minimize the complexity of the range polygon, \ie, its
number of segments.  We note that using $G_p$ may require more holes
than $G$ (see also Figure~\ref{fig:crossing_edges}d), but guarantees the existence of a solution.

Consider the graph $\graph'$ consisting of the union of the reachable
and the unreachable graph.  Clearly, all segments of the range
polygon lie in faces of $\graph'$.  A face of $\graph'$ that is
incident to both reachable and unreachable components is called
\emph{border region}.  Since a range polygon separates the reachable
and unreachable parts, each border region contains at least one
connected component of a range polygon.  Therefore, the number of
border regions is a lower bound on the number of holes.  On the other
hand, components in faces that are not border regions can be removed and several connected components in the
same border region can always be merged, potentially at
the cost of increasing the complexity; see Figure~\ref{fig:holes}.  Therefore, a range polygon
with the minimum number of holes (with respect to $\graph_\planar$)
can be computed as follows.
\begin{compactenum}
\item Compute the reachable and unreachable parts of~$\graph$.
\item Planarize $\graph$, compute the reachable and unreachable parts of
  the planarization~$\graph_\planar$.
\item Compute the border regions.
\item For each border region~$B$, compute a simple polygon of minimum
  complexity that is contained in $B$ and separates the unreachable
  components incident to $B$ from the reachable component.
\end{compactenum}
In the following sections we discuss several alternative
implementations for these steps.  The first three steps are described
together in Section~\ref{sec:range_query}.
\begin{figure}[t]
  \centering
  \includegraphics[page=1]{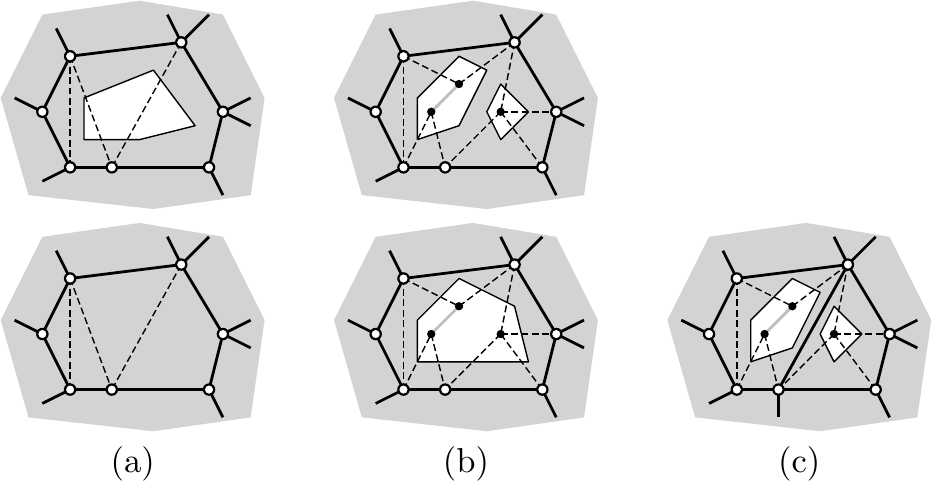}
  \caption{(a) A hole of the range polygon that contains no unreachable
    vertices can always be removed.  (b) Two holes that
    can be merged into one as they lie in the same border region.  (c) These two holes cannot be merged as they are separated by passable edges.}
  \label{fig:holes}
\end{figure}
The main part of the paper is concerned with
Step 4.  Each connected component of the boundary of a border region
is a hole-free non-crossing polygon.  Note that these polygons are not
necessarily simple in the sense that they may contain the same segment
twice in different directions; see Figure~\ref{fig:border_region}.
Thus, each border region is defined by two sets $R$ and $U$ of
hole-free non-crossing polygons, where $R$ contains the boundaries of the
reachable components and $U$ contains the boundaries of the unreachable
components.  For our range polygon, we seek a simple polygon with the
minimum number of links that separates $U$ from~$R$.  This problem has
been previously studied.  Guibas et al.~\cite{Gui93} showed that the
problem is \NP-complete in general.  In our case, however, it is
$|R| = 1$ since the reachable part of the network is, by definition,
connected.  Guibas et al.~\cite{Gui93} left this case as an open
problem, and, to the best of our knowledge, it has not been resolved.

\begin{figure}[t]
  \centering
  \includegraphics[page=1]{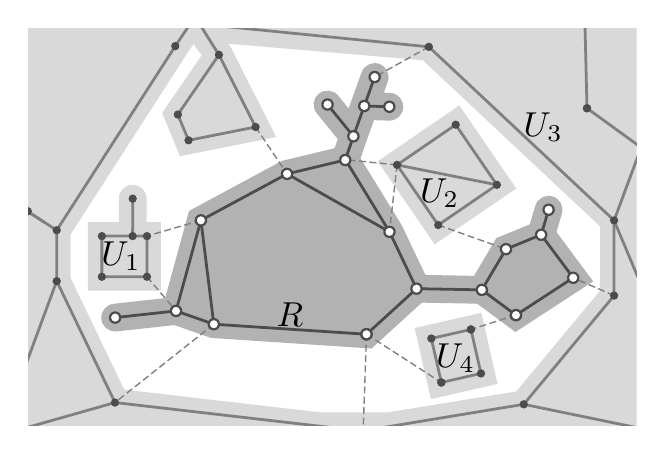}
  \caption{An example of a border region where the set~$U$ consists of several unreachable components with boundaries~$U_1$,~$U_2$,~$U_3$ and~$U_4$. Shaded areas show the reachable (dark gray) and unreachable (light gray) part with boundaries~$R$ and~$U$ of~$\graph$, respectively. The white area represents the border region.}
  \label{fig:border_region}
\end{figure}

In Section~\ref{sec:min_link_path}, we first consider border regions
that are incident to only one unreachable components, i.e.,
$|R| = |U| = 1$.  In this case, a polygon with the minimum number of
segments that separates $R$ and $U$ can be found in $O(n \log n)$
time (where $n$ is the total number of segments in the border region) using
the algorithm of Wang~\cite{Wan91}.  However,
this algorithm is rather involved, as it requires a (constant) number
of calls of a minimum-link path algorithm.  Instead, we propose a
simpler algorithm that uses at most two more segments than the
optimum, runs in linear time, and relies on a single call of a
minimum-link path algorithm.  In addition, we give a new linear
minimum link path algorithm that is simpler than previous algorithms
for this problem.
In Section~\ref{sec:heuristics}, we consider the general case of our
problem, where a border region may be incident to more than one
unreachable component.  We propose several algorithms for this
problem.  As mentioned above, the complexity of this problem is
unknown, and we therefore focus on efficient heuristic approaches that
work well in practice but do not necessarily give provable guarantees
on the number of segments.
%

\section{Computing the Border Regions}
\label{sec:range_query}

We describe an algorithm to compute the reachable and unreachable part of the input graph~$\graph = (\vertices, \edges)$, the planarization of these parts, and extraction of all border regions.
We modify Dijkstra's algorithm~\cite{d-ntpcg-59} to compute the reachable and unreachable parts in time $\bigO(|\vertices| \log |\vertices|)$, given that~$|\edges| \in \bigO(|\vertices|)$ for graphs representing road networks. Afterwards, we map this information to the planarized graph~$\graph_\planar$ in~$\bigO(|E|)$ time. Finally, we extract the border regions by traversing all faces of the planar graph that contain at least one boundary edge or accessible edge, which requires linear time in the size of all border regions. Below, we discuss efficient implementations of these three steps.
First, we describe a variant of Dijkstra's algorithm that computes, for a given source vertex~$\source$ and a range~$\range \in \posreals$
\begin{itemize}
\item for each vertex $\vertex \in \vertices$, whether~$\vertex$ is reachable from~$\source$;
\item the set~$\edges_x := \edges_b \cup \edges_a$ of all boundary edges and accessible edges;
\item for every edge $(\vertexa,\vertexb) \in \edges \setminus \edges_x$, whether $(\vertexa, \vertexb)$ is passable or unreachable.
\end{itemize}
For the sake of simplicity, we assume that shortest paths are unique
(\wrt the length function). Thereby, we avoid the special case of two
paths with equal distance but different consumption, which requires
additional tie breaking.
The algorithm works as follows.
Along the lines of Dijkstra's algorithm, it maintains vertex labels
consisting of (tentative) values for distance $\distancelabel(\cdot)$
and resource consumption $\consumptionlabel(\cdot)$, both initially
set to $0$ for~$\source$ and~$\infty$ for all other
vertices. The algorithm uses a priority queue of
vertices, initially containing~$\source$. In each step, it extracts
the vertex~$\vertexa$ with minimum distance label from the queue,
thereby \emph{settling} it. Then, all outgoing edges
$(\vertexa,\vertexb)$ are scanned, checking whether
$\distancelabel(\vertexa) + \weight(\vertexa,\vertexb) <
\distancelabel(\vertexb)$. If this is the case, the labels
at~$\vertexb$ are updated accordingly to
$\distancelabel(\vertexb) := \distancelabel(\vertexa) +
\weight(\vertexa,\vertexb)$ and
$\consumptionlabel(\vertexb) := \consumptionlabel(\vertexa) +
\cons(\vertexa,\vertexb)$. Also,~$\vertexb$ is inserted into the queue
if it is not contained already. Once the queue runs empty, we know
that a vertex $\vertexb \in \vertices$ is reachable if and only if
$\consumptionlabel(\vertexb) \le \range$.  An edge
$(\vertexa,\vertexb) \in \edges \setminus \edges_x$ is passable if~$\vertexa$ (and thus, also~$\vertexb$) is reachable, and it is unreachable otherwise.
Correctness follows directly from the correctness of Dijktra's algorithm and the fact that we simply sum up consumption values along shortest paths.

We describe how to compute the set~$\edges_x$ of boundary edges and accessible edges, \ie, all edges that intersect the interior of the border regions, which are required as input for later steps. A na\"ive approach could run a linear sweep
over all edges after the search terminates.  In practice, we can do better by computing~$\edges_x$ on-the-fly (especially when the stopping criterion described below is applied). We know
whether an edge~$(\vertexa,\vertexb)$ belongs to~$\edges_x$ as
soon as both~$\vertexa$ and~$\vertexb$ were settled (and thus have
final labels). Therefore, after extracting a vertex from the queue, we
check all incident edges and add them to~$\edges_x$ if the
respective neighbor was settled and one of the following conditions holds. Either, exactly one endpoint of the edge is currently reachable, or both endpoints are reachable but the edge is not passable, \ie, $\consumptionlabel(\vertexa) + \cons((\vertexa,\vertexb)) > \range$ and $\consumptionlabel(\vertexb) + \cons((\vertexb,\vertexa)) > \range$ if $(\vertexb,\vertexa) \in \edges$.

\paragraph{Stopping Criterion.}

Without further modification, the described algorithm settles all vertices in the graph (presuming it is strongly connected). In practice, this may be undesirable, especially for small ranges. However, simply pruning the search at unreachable vertices does not preserve correctness. Assume we do not insert an unreachable vertex~$\vertexa$ into the queue. Consider another unreachable vertex~$\vertexb$, such that the shortest $\source$--$\vertexb$-path contains $\vertexa$. There might exist some (non-shortest) path from $\source$ to $\vertexb$ with lower resource consumption that is found by the algorithm instead. Then, $\vertexb$ is falsely identified as reachable. However, we can safely abort the query as soon as no reachable vertex is left in the queue, since no reachable vertex can be found from an unreachable vertex (in case of negative consumption values, we presume that battery constraints apply~\cite{Bau13}). To efficiently check whether this stopping criterion is fulfilled, we simply maintain a counter to keep track of the number of reachable vertices in the queue.

Using the stopping criterion, we may not have found all boundary edges once the search terminates, as possibly not all unreachable endpoints of boundary edges have been settled. However, we know that the unreachable endpoint~$\vertexb$ of missing boundary edges~$(\vertexa,\vertexb)$ must have been added to the queue when~$\vertexa$ was settled.
We ensure that unreachable vertices~$\vertexa$ of boundary edges~$(\vertexa,\vertexb)$ are in the queue as well, by scanning incident incoming edges when settling the reachable vertex~$\vertexa$ and adding each neighbor to the queue with key~$\infty$ if it is not contained yet.
We obtain the remaining boundary edges in a sweep over all vertices left in the queue, checking for each its incident edges.

Even with the stopping criterion in place, the running time is impractical on large inputs. Speedup techniques~\cite{Bas14} are a common approach for much faster queries in practice. A recent work~\cite{bbdw-fcirn-15} presents techniques extending the basic approach described above that enable fast computation of boundary edges (considering the special case of $\weight \equiv \cons$). Since they can be adapted to our scenario, we focus on the remaining steps for computing range polygons.

\paragraph{Planarization.}


We planarize~$\graph$ in a preprocessing step to obtain~$\graph_\planar = (\vertices_\planar, \edges_\planar)$. This can be done using the well-known sweep line algorithm~\cite{Bentley:1979:ARC:1311099.1311743,Berg:2008:CGA:1370949}.
%
%
For each vertex in~$\graph_\planar$, we store its original vertex in~$\graph$ (if it exists).
In practice, where vertices are represented by indices~$\{1, \dots, |\vertices|\}$, this mapping can be done implicitly since $\vertices_\planar$ is a superset of~$\vertices$.
%
%

After computing the reachable subgraph of~$\graph$ as described above, we compute reachability of dummy vertices and the set~$\edges_x$ of boundary edges and accessible edges in~$\graph_\planar$ as follows.
First, we have to ensure that boundary edges and accessible edges returned by the algorithm are actually contained in~$\graph_\planar$. We add a flag to all edges in~$\graph$ during preprocessing that indicates whether an edge is contained in~$\graph_\planar$. Then, we modify our search algorithm to add edges to~$\edges_x$ only if this flag is set.
After the search terminates, we check for each dummy vertex~$\vertex \in \vertices_\planar$ whether it is reachable, by checking passability of all original edges in~$\edges$ that contain~$\vertex$. To this end, we precompute an array of all original edges that were split during planarization, and for each split edge a list of dummy vertices it contains (an original edge may intersect several other edges). Then, we can sweep over this array of edges, and for each passable edge, we mark all its dummy vertices as reachable.
Finally, we perform a sweep over all edges in~$\graph_\planar$ that have at least one dummy vertex as endpoint, to determine any missing edges in~$\edges_x$ (to test whether an edge is accessible, we need a pointer to the unique original edge containing it).
Note that the linear sweep steps produce limited overhead in practice,
since the number of dummy vertices in graphs representing road
networks is typically small (as large parts of the input are planar to
begin with).
Afterwards, a vertex in~$\vertices_\planar$ is reachable if it is a dummy vertex marked as reachable, or its corresponding original vertex is reachable (if it exists). Otherwise, it is unreachable.
For edges in~$\edges_\planar \setminus \edges_x$, we check reachability of their endpoints to determine whether they are passable or unreachable.

An alternative approach modifies the search algorithm to work directly
on the planar graph~$\graph_\planar$ to avoid the additional linear
sweeps. However, this produces overhead during the search (\eg, case
distinctions for dummy vertices). Consequently, such approaches
did not provide significant speedup in
preliminary experiments. Moreover, determining the reachable subgraph
of~$\graph_\planar$ in a separate step simplifies the integration of
speedup techniques for the more expensive
search algorithm~\cite{bbdw-fcirn-15}.

\paragraph{Extracting the Border Regions.}

Given the set~$\edges_x$ of boundary edges and accessible edges in~$\graph_\planar$, we describe how to compute the actual border regions (\ie, the polygons describing $R$ and~$U$).
The basic idea is to traverse all faces of $\graph_\planar$ that contain edges in~$\edges_x$, to collect the segments that form boundaries of the border regions. Clearly, all passable edges in these faces are part of some reachable boundary, while all unreachable edges belong to an unreachable boundary. Moreover, since~$\graph_\planar$ is strongly connected, all faces contained in a border region must contain an edge in~$\edges_x$. Thus, traversing these faces is sufficient to obtain all border regions.

\begin{figure}[tb!]
\centering
  \includegraphics[page=1]{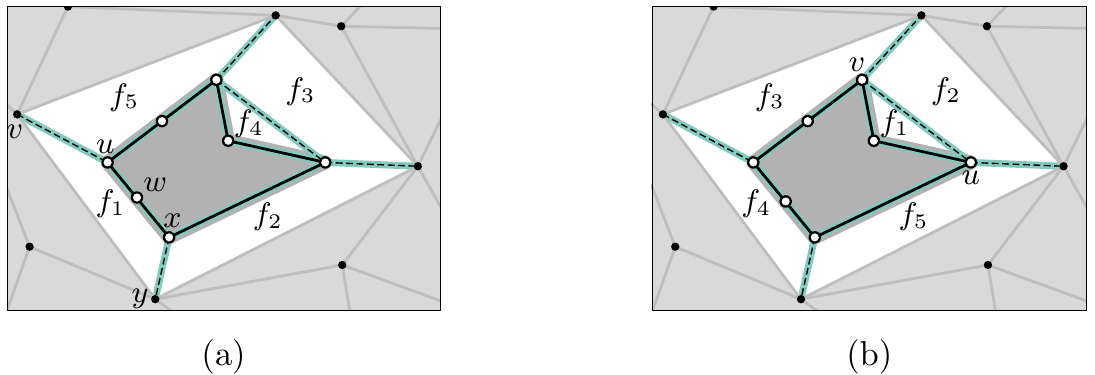}
  \caption{Visited edges (green) when extracting a reachable boundary. (a)~Starting at the boundary edge $(u,v)\in \edges_x$, the face~$f_1$ is traversed first until the boundary edge~$(x, y)\in \edges_x$ is encountered. Afterwards, the faces~$f_2,f_3,f_4,f_3,f_5$ are visited in this order until~$(u,v)$ is reached again. (b)~Starting at the accessible edge~$(u,v) \in \edges_x$, the face~$f_1$ is traversed until~$(v,u) \in \edges_x$ is reached. Then, the faces~$f_2,f_3,f_4,f_5,f_2$ are processed before~$(u,v)$ is reached again.}
  \label{fig:border-region-extraction}
\end{figure}

In somewhat more detail, we maintain two flags for every edge~$(u,v)$ in~$\edges_x$ indicating
whether~$u$ or~$v$ has been visited, respectively, initially set to
false.
Let~$(u,v)$ be the first edge of~$\edges_x$
that is considered, and without loss of generality, let~$u$ be
reachable.
We compute the (unique) reachable component~$R_{(u,v)}$ of the border
region~$B_{(u,v)}$ containing~$(u,v)$; see Figure~\ref{fig:border-region-extraction}. We mark~$(u,v)$ as visited and traverse the face left of~$(u,v)$, following the unique neighbor~$w$ of~$u$ in this
face that is not~$v$. Every edge that we traverse is added to~$R_{(u,v)}$.
As soon as we encounter an edge~$(x,y) \in \edges_x$,
we continue by traversing the \emph{twin face} of~$(x,y)$, \ie, the unique face of $\graph_\planar$ that contains the other side of~$(x,y)$. The edge~$(x,y)$ itself
is not added to~$R_{(u,v)}$. Moreover, we mark~$(x,y)$ as
visited. The current extraction step is finished as soon as~$(u,v)$ with the same orientation is reached again; see Figure~\ref{fig:border-region-extraction}.
%
%
If~$v$ is unreachable, $(u,v)$ is a boundary edge. Thus, we continue with the extraction of the unreachable component~$U_{(u,v)}$ containing~$v$ in the same manner and assign it to~$B_{(u,v)}$.

We loop over the remaining edges in~$\edges_x$ and extract boundaries corresponding to vertices not visited before.
By extracting reachable components first, we ensure that the corresponding reachable boundary of some unreachable component is always known before extraction, namely, the boundary containing the reachable endpoint of the considered edge in~$\edges_x$.
Therefore, the unreachable component is assigned to the unique border region that contains this reachable boundary.

\subparagraph{Implementation Details.}

To extract the components of all border regions, we have to traverse the faces of the planar input graph. We use a cache-friendly data structure to represent these faces, allowing us to run along a face efficiently. For each face, we store the sequence of vertices as they are found traversing the face in clockwise order starting at an arbitrary vertex. At the beginning and at the end of this sequence, we store sentinels that hold the index of the first and last entry of a vertex of the corresponding face, respectively. Then, we can use one single array that holds all faces of the graph.
Traversing the face in either direction requires only a sweep, jumping at most once to the beginning or end of the face, respectively.
For consecutive vertices~$\vertexa,\vertexb$ in this array, we store at~$\vertexb$ its index in the corresponding twin face of the edge~$(\vertexa,\vertexb)$. Finally, we store for every edge~$(\vertexa,\vertexb)$ in the graph the two indices of the head vertex~$\vertexb$ in this data structure, \ie, its occurrence in the faces left and right of this edge.

To efficiently decide whether an edge is in~$\edges_x$, or if an edge in~$\edges_x$ was marked as visited, we store this list as an array and sort it (\eg, by the index of the head vertex) before extracting the reachable components. Then we can quickly retrieve an edge in this array using binary search (we also tried using hash sets as an alternative approach, but this turned out to be slightly slower in preliminary experiments).

\section{Range Polygons in Border Regions Without Holes}
\label{sec:min_link_path}

Given a border region~$B$ with reachable component~$R$ and a single unreachable component~$U$, we present an algorithm for computing a polygon that separates~$R$ and~$U$. In Section~\ref{sec:heuristics}, we describe how to generalize our approach to the case~$|U| > 1$.

The basic idea is to add an arbitrary boundary edge~$\edge$ to~$B$, thereby connecting both components~$R$ and~$U$. Since we presume that~$\graph$ is strongly connected, such a boundary edge always exists. In
the resulting hole-free non-crossing polygon~$B'$, we compute a path with
minimum number of segments that connects both sides of~$\edge$.
The algorithm of Suri~\cite{Sur86} computes such a \emph{minimum-link path}~$\apath'$ in linear time. We obtain a separating polygon~$S'$ by connecting the endpoints of~$\apath'$ along~$\edge$. It is easy to see that this yields a polygon with at most two additional segments compared to an optimal solution.
\begin{figure}
\centering
\includegraphics[page=1]{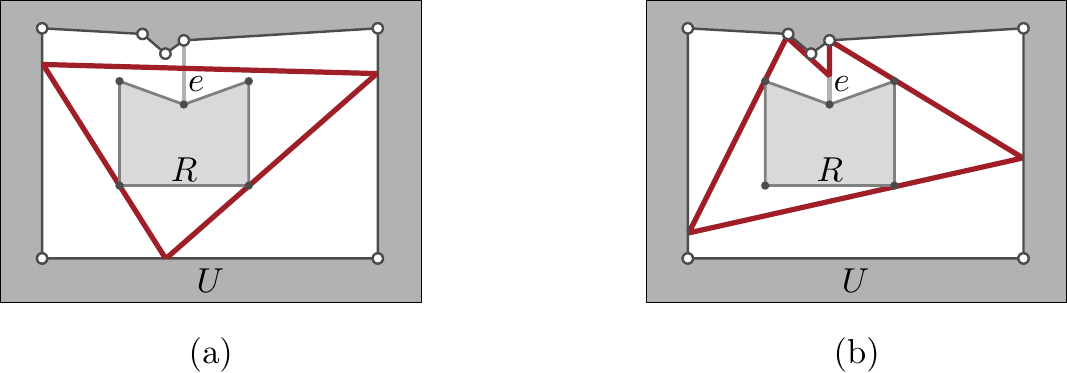}
\caption{(a)~The polygon $S$ (red) separating $R$ and $U$ with $\opt = 3$ links.  (b)~The polygon $S'$ obtained by a minimum-link path from~$e$ has $\opt+2 = 5$ segments.}
\label{fig:approximation-lemma}
\end{figure}
\begin{lemma}
  Let $S$ be a polygon that separates~$R$ and~$U$ with minimum
  number of segments, and let~$\opt$ denote this number. Then~$S'$ has
  at most $\opt + 2$ segments.
\end{lemma} 

\begin{proof}
 We can split~$S$ at~$\edge$ into a path~$\apath$ connecting both sides of~$\edge$. Clearly,~$\apath$ has at most~$\opt + 1$ links (if a segment of~$S$ crosses~$\edge$, we split it into two segments with endpoints in~$\edge$; see Figure~\ref{fig:approximation-lemma}). Since~$\apath'$ is a minimum-link path, we have~$|\apath'| \le |\apath| = \opt + 1$. Moreover,~$S'$ is obtained by adding a single subsegment of~$\edge$ to~$\apath'$, so its complexity is bounded by~$\opt + 2$.
\end{proof}

We now address the subproblem of
computing a minimum-link path between two edges of a simple polygon.
%
%
The linear-time algorithm of Suri~\cite{Sur86} starts by triangulating the input polygon. To save running time, we can preprocess this step and triangulate all faces of the planar graph~$\graph_\planar$.
Afterwards, in each step of Suri's algorithm, a \emph{window} (which we formally
define in a moment) is computed.
To obtain the windows in linear time, it relies on several calls to a subroutine computing visibility polygons. While this is sufficient to prove linear running time, it seems wasteful from a practical point of view.
In the following, we present an
alternative algorithm for computing the windows that also results in
linear running time, but is much simpler.  It can be seen as a generalization of
an algorithm by Imai and Iri~\cite{ii-a-87} for approximating
piecewise linear functions.

\paragraph{Windows and Visibility.}

Let $P$ be a simple polygon and let $a$ and $b$ be edges of~$P$.  We
want to compute a minimum-link polygonal path starting at $a$ and
ending in $b$ that lies in the interior of $P$.  Let $T$ be the graph obtained by arbitrarily
triangulating~$P$.  Let $t_a$ and $t_b$ be the triangles incident to
$a$ and $b$, respectively.  As $T$ is an outerplanar graph, its (weak)
dual graph has a unique path
$t_a = t_1, t_2, \dots, t_{k-1}, t_k = t_b$ from $t_a$ to $t_b$; see
Figure~\ref{fig:visibility-notation-1}a.  We call the triangles on
this path \emph{important} and their position in the path their
\emph{index}.

The \emph{visibility polygon} $V(a)$ of the edge $a$ in $P$ is the
polygon that contains a point $p$ in its interior if and only if there
is a point $q$ on $a$ such that the line segment $pq$ lies inside~$P$.
Let $i$ be the highest index such that the intersection of the
triangle $t_i$ with the visibility polygon $V(a)$ is not empty.  The
\emph{window} $w(a)$ is the edge of $V(a)$ that intersects $t_i$
closest to the edge between $t_i$ and $t_{i+1}$; see
Figure~\ref{fig:visibility-notation-1}b.  Note that $w(a)$ separates the
polygon $P$ into two parts.  Let $P'$ be the part containing the edge
$b$ that we want to reach.  A minimum-link path from $a$ to $b$ in $P$ can
then be obtained by adding an edge from $a$ to $w(a)$ to a minimum-link path from $w(a)$
to $b$ in $P'$.  Thus, the next window is computed in $P'$ starting
with the previous window $w(a)$.  In the following, we first describe
how to compute the first window and then discuss what has to be
changed to compute the subsequent windows.

\begin{figure}[t]
  \centering
  \includegraphics[page=1]{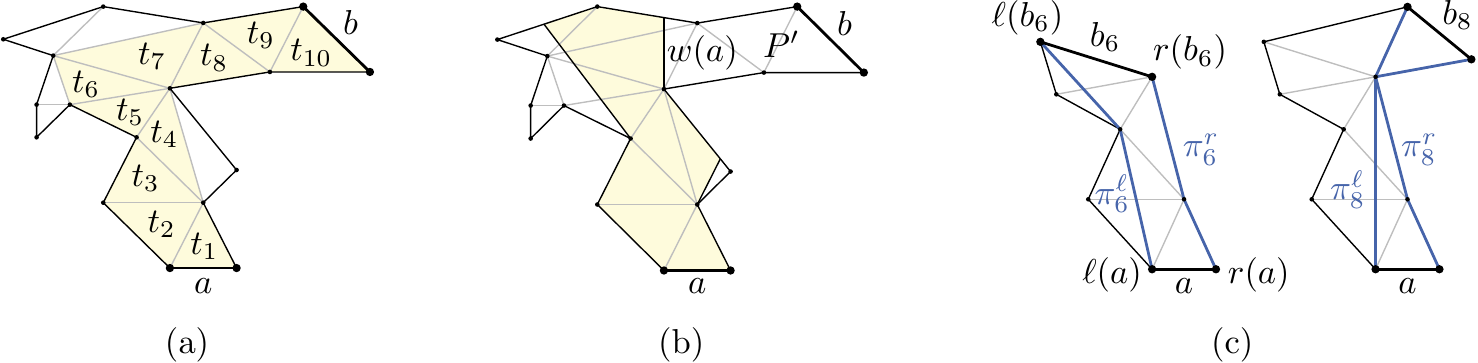}
  \caption{(a)~The important triangles (with respect to $a$ and $b$)
    of a polygon.  (b)~The window $w(a)$ is an edge of the (shaded)
    visibility polygon.  (c)~The left and right shortest paths (blue)
    intersect for $i = 8$ but not for $i = 6$.}
  \label{fig:visibility-notation-1}
\end{figure}

Let $T_i$ be the subgraph of $T$ induced by the triangles $t_1, \dots,
t_i$ and let $P_i$ be the polygon bounding the outer face of $T_i$.
The polygon $P_i$ has two special edges, namely $a$ and the edge
shared by $t_i$ and $t_{i+1}$, which we call $b_i$.  Let $\lef(a)$ and
$\rig(a)$, and $\lef(b_i)$ and $\rig(b_i)$ be the endpoints of $a$ and
$b_i$, respectively, such that their clockwise order is $\rig(a)$,
$\lef(a)$, $\lef(b_i)$, $\rig(b_i)$ (think of $\lef(\cdot)$ and
$\rig(\cdot)$ being the left and right endpoints, respectively); see
Figure~\ref{fig:visibility-notation-1}c.  We define the \emph{left
  shortest path} $\pi_i^{\lef}$ to be the shortest polygonal path
(shortest in terms of Euclidean length) that connects $\lef(a)$ with
$\lef(b_i)$ and lies inside or on the boundary of $P_i$.  The
\emph{right shortest path} $\pi_i^{\rig}$ is defined analogously for
$\rig(a)$ and $\rig(b_i)$; see Figure~\ref{fig:visibility-notation-1}c.

Assume that the edge $b_i$ is visible from $a$, i.e., there exists a
line segment in the interior of $P_i$ that starts at $a$ and ends at
$b_i$.  Such a visibility line separates the polygon into a left and a
right part.  Observe that it follows from the triangle inequality that
the left shortest path $\pi_i^{\lef}$ and the right shortest path
$\pi_i^{\rig}$ lie inside the left and right part, respectively.
Thus, these two paths do not intersect.  Moreover, the two shortest
paths are \emph{outward convex} in the sense that the left shortest
paths $\pi_i^{\lef}$ has only left bends when traversing it from
$\lef(a)$ to $\lef(b_i)$ (the symmetric property holds for
$\pi_i^{\rig}$); see the case $i = 6$ in
Figure~\ref{fig:visibility-notation-1}c.  We note that the outward
convex paths are sometimes also called ``inward convex'' and the
polygon consisting of the two outward convex paths together with the
edges $a$ and $b_i$ is also called \emph{hourglass}~\cite{gh-ospqsp-89}.  The following lemma, which
is similar to a statement shown by Guibas~et~al.~\cite[Lemma~3.1]{Gui87}, summarizes the above observation.

\begin{lemma}
  \label{lem:visible-impl-no-intersection}
  If the triangle $t_{i}$ is visible from $a$, then the left and right
  shortest paths in $P_{i-1}$ have empty intersection.  Moreover, if
  these paths do not intersect, they are outward convex.
\end{lemma}

\begin{figure}[t]
  \centering
  \includegraphics[page=1]{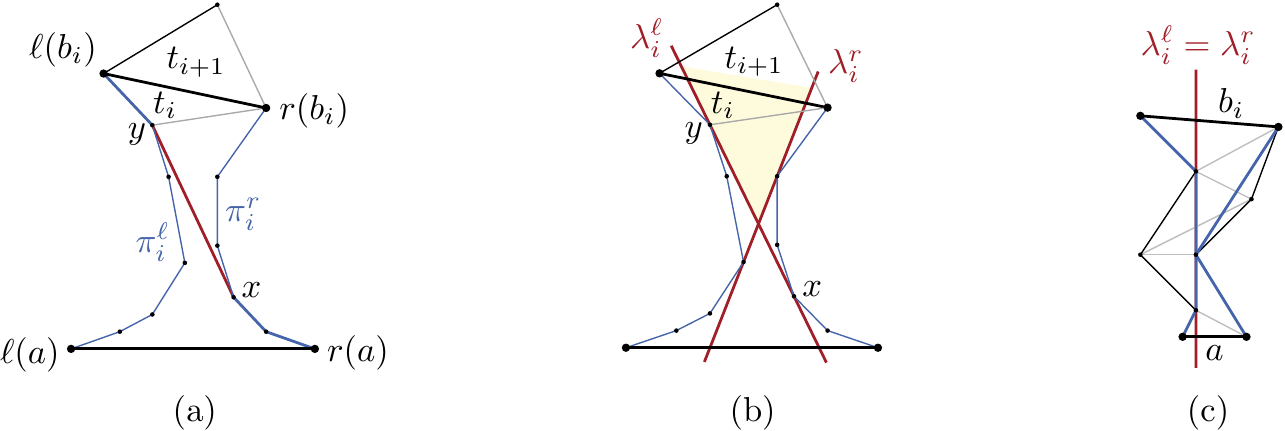}
  \caption{(a)~The shortest path from $\rig(a)$ to $\lef(b_i)$
    consists of the bold prefix of $\pi_i^{\rig}$, the red segment,
    and the bold suffix of $\pi_i^{\lef}$.  (b)~The two visibility
    lines spanning the (shaded) visibility cone.  (c)~Degenerated
    case.}
  \label{fig:hourglass}
\end{figure}

Guibas et al.~\cite{Gui87} argue that the converse of the first
statement is also true, i.e., if the two paths have empty
intersection, then the triangle $t_{i+1}$ is visible from~$a$.  Their
main arguments go as follows.  The shortest path (\wrt Euclidean
length) in the hourglass that connects $\rig(a)$ with $\lef(b_i)$ is
the concatenation of a prefix of $\pi_i^{\rig}$, a line segment from a
vertex $x$ of $\pi_i^{\rig}$ to a vertex $y$ of $\pi_i^{\lef}$, and a
suffix of $\pi_i^{\lef}$; see Figure~\ref{fig:hourglass}a.  We call the
straight line through $x$ and $y$ the \emph{left visibility line} and
denote it by $\lambda_i^{\lef}$.  We assume $\lambda_i^{\lef}$ to be
oriented from $x$ to $y$ and call $x$ and $y$ the \emph{source} and
\emph{target} of $\lambda_i^{\lef}$.  Analogously, one can define the
\emph{right visibility line} $\lambda_i^{\rig}$; see
Figure~\ref{fig:hourglass}b.  We call the intersection of the half-plane
to the right of $\lambda_i^{\lef}$ with the half-plane to the left of
$\lambda_i^{\rig}$ the \emph{visibility cone}.  It follows that the
intersection of the visibility cone with the edge $b_i$ is not empty
and a point on the edge $b_i$ is visible from $a$ if and only if it
lies in this intersection~\cite{Gui87}.  This directly extends to the
following lemma.

\begin{lemma}
  \label{lem:no-intersection-impl-visible}
  If the left and right shortest paths in $P_{i-1}$ have empty
  intersection, $t_{i}$ is visible from $a$.  Moreover, a point
  in $t_{i}$ is visible from $a$ if and only if it lies in the
  visibility cone.
\end{lemma}

Note that when~$\pi_i^{\rig}$ and~$\pi_i^{\lef}$ are both outward
convex and intersect, the visibility cone degenerates to a line; see
Figure~\ref{fig:hourglass}c. For the sake of simplicity, we presume
that all points are in general position, which prevents this special
case. (In practice, it is handled implicitly by the implementation of
Line~\ref{line:return-window-test} in Algorithm~\ref{alg:first-window}
below.)
The above observations then justify the following approach for computing
the window.  We iteratively increase $i$ until the left and the right
shortest path of the polygon $P_i$ intersect.  We then know that the
triangle $t_{i+1}$ is no longer visible; see
Lemma~\ref{lem:visible-impl-no-intersection}.  Moreover, as the left
and the right shortest paths did not intersect in $P_{i-1}$, the
triangle $t_i$ is visible from~$a$; see
Lemma~\ref{lem:no-intersection-impl-visible}.  To find the window, it
remains to find the edge of the visibility polygon $V(a)$ that
intersects $t_i$ closest to the edge between $t_i$ and $t_{i+1}$.
Thus, by the second statement of
Lemma~\ref{lem:no-intersection-impl-visible}, the window must be a
segment of one of the two visibility lines.
It remains to fill out the details of this algorithm, argue that it
runs in overall linear time, and describe what has to be done in later
steps, when we start at a window instead of an edge.

\begin{figure}[t]
  \centering
  \includegraphics[page=1]{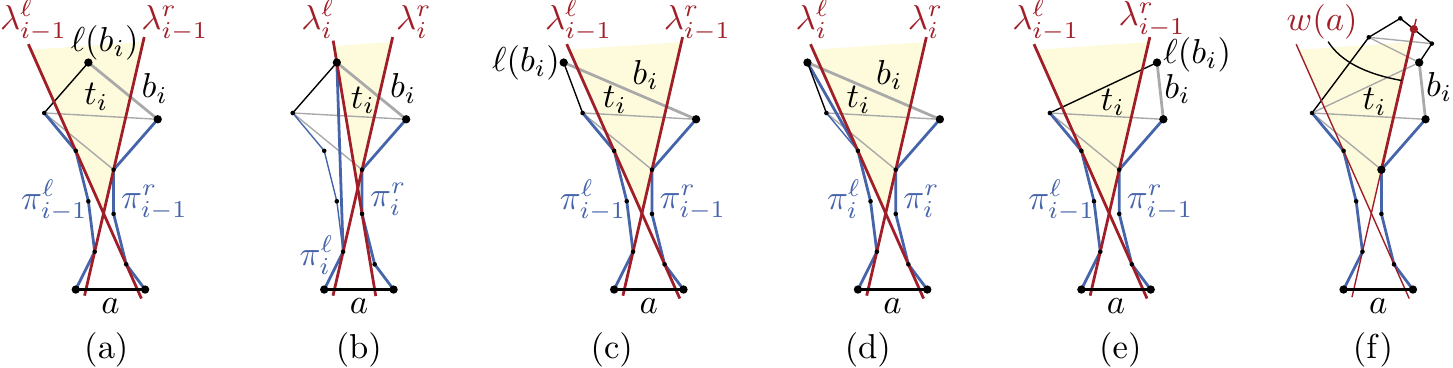}
  \caption{(a)~The new vertex $\lef(b_i)$ lies in the visibility cone.
    (b)~The updated left shortest path $\pi_i^{\lef}$ and left
    visibility line $\lambda_i^{\lef}$.  (c)~The vertex $\lef(b_i)$
    lies to the left of $\lambda_{i-1}^{\lef}$.  (d)~The left shortest
    path has to be updated, the left visibility line remains
    unchanged.  (e)~The vertex $\lef(b_i)$ lies to the right of
    $\lambda_{i-1}^{\rig}$, i.e., $t_{i+1}$ is not visible from $a$.
    (f)~The window $w(a)$ is a segment of $\lambda_{i-1}^{\rig}$.}
  \label{fig:first-window-new-triangle}
\end{figure}

\paragraph{Computing the First Window.}

We start with the details of the algorithm starting from an edge; also
see Algorithm~\ref{alg:first-window}.  Assume the triangle $t_i$ is
still visible from $a$, i.e., $\pi_{i-1}^{\lef}$ and
$\pi_{i-1}^{\rig}$ do not intersect.  Assume further that we have
computed the left and right shortest paths $\pi_{i-1}^{\lef}$ and
$\pi_{i-1}^{\rig}$ as well as the corresponding visibility lines
$\lambda_{i-1}^{\lef}$ and $\lambda_{i-1}^{\rig}$ in a previous step.
Assume without loss of generality that the three corners of the
triangle $t_i$ are $\lef(b_{i-1})$, $\lef(b_i)$, and
$\rig(b_i) = \rig(b_{i-1})$.
%
%
There are three
possibilities shown in Figure~\ref{fig:first-window-new-triangle}, i.e., the
new vertex $\lef(b_i)$ lies either in the visibility cone spanned by
$\lambda_{i-1}^{\lef}$ and $\lambda_{i-1}^{\rig}$, to the left of the
left visibility line $\lambda_{i-1}^{\lef}$, or to the right of the
right visibility line $\lambda_{i-1}^{\rig}$.

By Lemma~\ref{lem:no-intersection-impl-visible}, a point in $t_i$ is
visible from $a$ if and only if it lies inside the visibility cone.
Thus, the edge $b_i$ between $t_i$ and $t_{i+1}$ is no longer visible
if and only if the new vertex $\lef(b_i)$ lies to the right of
$\lambda_{i-1}^{\rig}$; see Figure~\ref{fig:first-window-new-triangle}e.
In this case, we can stop and the desired window $w(a)$ is the segment
of $\lambda_{i-1}^{\rig}$ starting at its touching point with
$\pi_{i-1}^{\rig}$ and ending at its first intersection with an edge
of $P$; see Figure~\ref{fig:first-window-new-triangle}f and
lines~\ref{line:return-window-start}--\ref{line:return-window-end} of
Algorithm~\ref{alg:first-window}.

\begin{algorithm}[t]
  \caption{\label{alg:first-window}Computes the first window $w(a)$.}%
  \newcommand{\commentfont}[1]{\texttt{\color{gray}#1}}
  \SetCommentSty{commentfont} \tcp{initial paths (one-vertex
    sequences) and visibility lines} $\pi^{\lef} = [\lef(a)]$;\quad%
  $\pi^{\rig} = [\rig(a)]$;\quad%
  $\lambda^{\lef} = \mathrm{line}(\rig(a), \lef(a))$;\quad%
  $\lambda^{\rig} = \mathrm{line}(\lef(a), \rig(a))$\;%
  \For{$i = 1$ \KwTo $k$}{%
    \eIf{$\rig(b_i) = \rig(b_{i-1})$}{%
      \tcp{$b_i$ not visible $\Rightarrow$ return window}%
      \If{$\lef(b_i)$ lies to the right of $\lambda^{\rig}$\label{line:return-window-test}}{%
        \label{line:return-window-start}%
        $x = $ first intersection of $\lambda^{\rig}$ with $P$ after
        $\mathrm{target}(\lambda^{\rig})$\;
        \KwRet{$\mathrm{segment}(\mathrm{target}(\lambda^{\rig}),
          x)$}%
        \label{line:return-window-end}%
      }%
      \tcp{extend left path $\pi^{\lef}$ like in Graham's scan}%
      append $\lef(b_i)$ to $\pi^{\lef}$\;%
      \label{line:extend-path-start}%
      \While{last bend of $\pi^{\lef}$ is a right bend}{%
        remove second to last element from $\pi^{\lef}$\;%
        \label{line:extend-path-end}%
      }%
      \tcp{$\lef(b_i)$ in visibility cone $\Rightarrow$ update left
        visibility line $\lambda^{\lef}$}%
      \If{$\lef(b_i)$ lies to the right of $\lambda^{\lef}$}{%
        \label{line:visibility-line-start}%
        $\mathrm{target}(\lambda^{\lef}) = \lef(b_i)$\;%
        \While{$\lambda^{\lef}$ is not a tangent of $\pi^{\rig}$ at
          $\mathrm{source}(\lambda^{\lef})$}{%
          $\mathrm{source}(\lambda^{\lef}) = $ successor of
          $\mathrm{source}(\lambda^{\lef})\text{ in }\pi^{\rig}$\;%
          \label{line:visibility-line-end}%
        }%
      }%
    }{%
      \tcp{case $\lef(b_i) = \lef(b_{i-1})$ is symmetric to above case
        $\rig(b_i) = \rig(b_{i-1})$}%
    }%
  }%
\end{algorithm}

In the other two cases (Figure~\ref{fig:first-window-new-triangle}a and
Figure~\ref{fig:first-window-new-triangle}c), we have to compute the new
left and right shortest paths $\pi_i^{\lef}$ and $\pi_i^{\rig}$ and
the new visibility lines $\lambda_i^{\lef}$ and $\lambda_i^{\rig}$
(Figure~\ref{fig:first-window-new-triangle}b and
Figure~\ref{fig:first-window-new-triangle}d).  Note that the old
and new right shortest paths $\pi_{i-1}^{\rig}$ and $\pi_i^{\rig}$
connect the same endpoints $\rig(a)$ and $\rig(b_{i-1}) = \rig(b_i)$.
As the path cannot become shorter by going through the new triangle
$t_i$, we have $\pi_{i}^{\rig} = \pi_{i-1}^{\rig}$.  The same argument
shows that $\lambda_{i}^{\rig} = \lambda_{i-1}^{\rig}$ (recall that
the visibility lines were defined using a shortest path from $\lef(a)$
to~$\rig(b_{i-1}) = \rig(b_{i})$).

We compute the new left shortest path $\pi_i^{\lef}$ as follows; see
Lines~\ref{line:extend-path-start}--\ref{line:extend-path-end} of
Algorithm~\ref{alg:first-window}.  Let $x$ be the latest vertex on
$\pi_{i-1}^{\lef}$ such that the prefix of $\pi_{i-1}^{\lef}$ ending
at $x$ concatenated with the segment from $x$ to $\lef(b_i)$ is
outward convex.  We claim that $\pi_i^{\lef}$ is the path obtained by
this concatenation, i.e., this path lies inside $P_i$ and there is no
shorter path lying inside~$P_i$.  It follows by the outward convexity,
that there cannot be a shorter path inside $P_i$ from $\lef(a)$ to~$\lef(b_i)$.  Moreover, by the assumption that $\pi_{i-1}^{\lef}$ was
the correct left shortest path in~$P_{i-1}$, the subpath from
$\lef(a)$ to $x$ lies inside $P_i$.  Assume for contradiction that the
new segment from $x$ to $\lef(b_i)$ does not lie entirely inside
$P_i$.  Then it has to intersect the right shortest path and it
follows that the right shortest path and the correct left shortest
path have non-empty intersection, which is not true by
Lemma~\ref{lem:visible-impl-no-intersection}.

\begin{figure}[t]
  \centering
  \includegraphics[page=2]{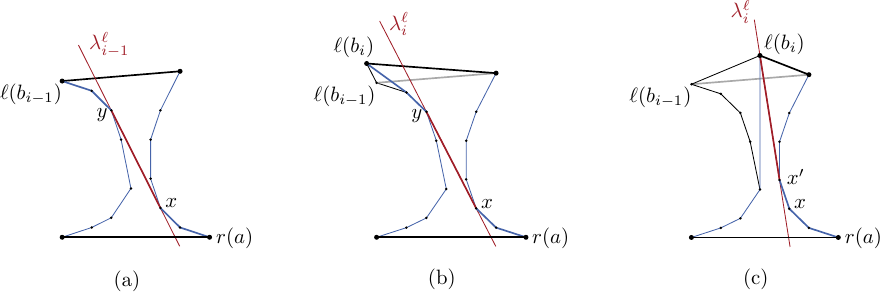}
  \caption{(a)~The shortest path from $\rig(a)$ to $\lef(b_{i-1})$
    (bold) defining the left visibility line~$\lambda_{i-1}^{\lef}$.
    (b)~The visibility line does not change if $\lef(b_i)$ lies to the
    left of $\lambda_{i-1}^{\lef}$.  (c)~Illustration how the
    visibility line changes when $\lef(b_i)$ lies to the right of~$\lambda_{i-1}^{\lef}$.}
  \label{fig:visibility-lines}
\end{figure}

To get the new left visibility line $\lambda_i^{\lef}$, we have to
consider the shortest path in $P_i$ that connects $\rig(a)$ with
$\lef(b_i)$.  Let $x$ and $y$ be the source and target of
$\lambda_{i-1}^{\lef}$, respectively, i.e., the shortest path from
$\rig(a)$ to $\lef(b_{i-1})$ is as shown in
Figure~\ref{fig:visibility-lines}a.  If the new vertex $\lef(b_i)$
lies to the left of $\lambda_{i-1}^{\lef}$
(Figure~\ref{fig:visibility-lines}b), then the shortest path from
$\rig(a)$ to $\lef(b_{i})$ also includes the segment from $x$ to $y$.
Thus, $\lambda_i^{\lef} = \lambda_{i-1}^{\lef}$ holds in this case.
Assume the new vertex $\lef(b_i)$ lies to the right of
$\lambda_{i-1}^{\lef}$ (Figure~\ref{fig:visibility-lines}c).  Let $x'$
be the latest vertex on the path $\pi_i^{\rig}$ such that the
concatenation of the subpath from $\rig(a)$ to $x'$ with the segment
from $x'$ to the new vertex $\lef(b_i)$ is outward convex in the sense
that it has only right bends; see Figure~\ref{fig:visibility-lines}c.
We claim that this path lies inside $P_i$ and that there is no shorter
path inside $P_i$.  Moreover, we claim that $x'$ is either a successor
of $x$ in $\pi_{i-1}^{\rig}$ or $x' = x$.  Clearly, the concatenation
of the path from $\rig(a)$ to $x$ with the segment from $x$ to
$\lef(b_i)$ is outward convex, thus the latter claim follows.  It
follows that the segment from $x'$ to $\lef(b_i)$ lies to the right of
the old visibility line~$\lambda_{i-1}^{\lef}$.  Thus, it cannot
intersect the path $\pi_i^{\lef}$ (except in its
endpoint~$\lef(b_i)$), as $\pi_{i-1}^{\lef}$ lies to the left
of~$\lambda_{i-1}^{\lef}$.  Moreover, as we chose $x'$ to be the last
vertex on $\pi_{i-1}^{\rig}$ with the above property, this new segment
does not intersect $\pi_i^{\rig}$ (except in $x'$).
Hence, the segment from $x'$ to $\lef(b_i)$ lies inside $P_i$.  As
before, it follows from the convexity that there is no shorter path
inside $P_i$.  Thus, $\lambda_i^{\lef}$ is the line through $x'$ and
$\lef(b_i)$ ($x'$ is the new source and $\lef(b_i)$ is the new
target).  It follows that
Lines~\ref{line:visibility-line-start}--\ref{line:visibility-line-end}
correctly compute the new left visibility line.

\begin{lemma}
  \label{lem:window-linear-time}
  Let $t_h$ be the triangle with the highest index that is visible
  from~$a$.  Then Algorithm~\ref{alg:first-window} computes the first
  window $w(a)$ in $O(h)$ time.
\end{lemma}
\begin{proof}
  We already argued that Algorithm~\ref{alg:first-window} correctly
  computes the first window.  To show that it runs in $O(h)$ time,
  first note that the polygon $P_h$ has linear size in $O(h)$.  Thus,
  it suffices to argue that the running time is linear in the size of
  $P_h$.  In each step $i$, we first check whether the next triangle
  is still visible by testing whether the new vertex $\lef(b_i)$ (or~$\rig(b_i)$) lies to the right of the visibility line
  $\lambda_{i-1}^{\rig}$ (or to the left of $\lambda_{i-1}^{\lef}$).
  This takes only constant time.  When updating the left and right
  shortest paths, we have to iteratively remove the last vertex of the
  previous path until the resulting path is outward convex.  This
  takes time linear in the number of vertices we remove.  However, a
  vertex removed in this way will never be part of a left or right
  shortest path again.  Thus, the number of these removal operations
  over all $h$ steps is bounded by the size of $P_h$.  When updating
  the visibility lines, the only operation potentially consuming more
  than constant time is finding the new source $x'$.  As $x'$ is a
  successor of the previous source $x$ (or $x' = x$), we never visit a vertex of $P_h$
  twice in this type of operation.  Thus, the total running time of
  finding these successors over all $h$ steps is again linear in the
  size of $P_h$.
\end{proof}

\paragraph{Initialization for Subsequent Windows.}

\begin{figure}[t]
  \centering
  \includegraphics[page=1]{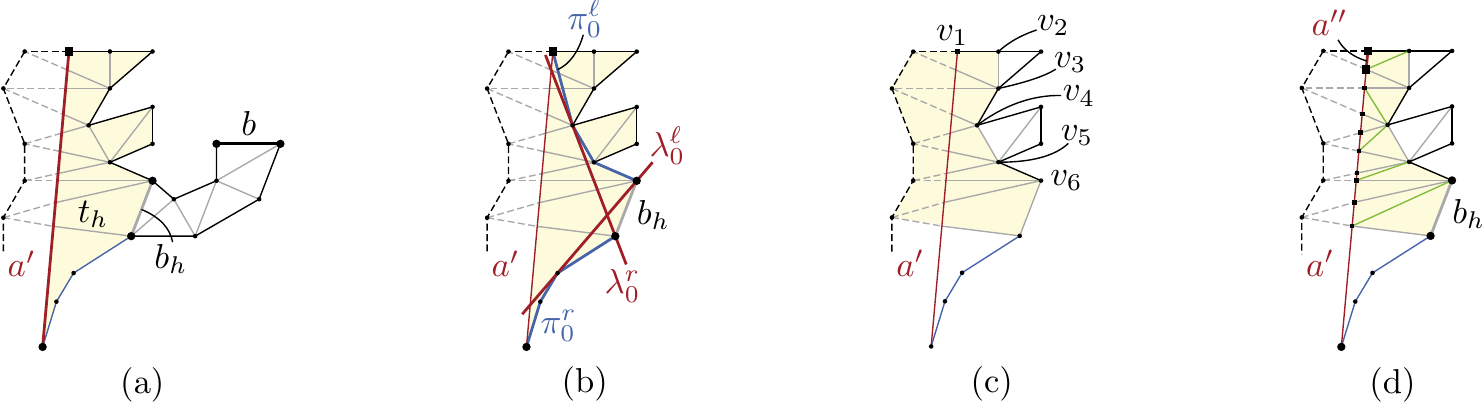}
  \caption{(a)~The polygon $P'$ we are interested in after computing
    the first window $a'$.  The initial part $P_0'$ is shaded yellow.
    (b)~Initial left and right shortest paths $\pi_0^{\lef}$ and
    $\pi_0^{\rig}$ (blue) with corresponding visibility lines
    (red).  (c)~The sequence $v_1, \dots, v_6$ we use for Graham's
    scan.  Triangles of $P$ intersected by $a'$ are shaded yellow.
    (d)~Computing the shortest path from $\lef(a'') = \lef(a')$ to
    $\lef(b_h)$ in this subdivided polygon using
    Algorithm~\ref{alg:first-window} actually applies Graham's scan to
    $v_1, \dots, v_6$.}
  \label{fig:second-window}
\end{figure}

As mentioned before, the first window $w(a)$ we compute separates $P$
into two smaller polygons.  Let $P'$ be the part including the edge~$b$
(and not $a$).  In the following, we denote $w(a)$ by $a'$.  To
get the next window $w(a')$, we have to apply the above procedure to
$P'$ starting with $a'$.  However, this would require to partially
retriangulate the polygon $P'$.  More precisely, let $t_h$ be the
triangle with the highest index that is visible from $a$ and let $b_h$
be the edge between $t_h$ and $t_{h+1}$; see
Figure~\ref{fig:second-window}a.  Then $b_h$ separates $P'$ into an
initial part $P_0'$ (the shaded part in Figure~\ref{fig:second-window}a)
and the rest (having $b$ on its boundary).  The latter part is
properly triangulated, however, the initial part $P_0'$ is not.  The
conceptually simplest solution is to retriangulate $P_0'$.
However, this would require an efficient subroutine for triangulation (and dynamic data structures that allow us to update~$P$ and~$T$, which produces overhead in practice).
Instead, we propose a much simpler method for computing the next window.

The general idea is to compute the shortest paths in $P_0'$ from
$\lef(a')$ to $\lef(b_h)$ and from $\rig(a')$ to $\rig(b_h)$; see
Figure~\ref{fig:second-window}b.  We denote these paths by
$\pi_0^{\lef}$ and $\pi_0^{\rig}$, respectively.  Moreover, we want to
compute the corresponding visibility lines $\lambda_0^{\lef}$ and
$\lambda_0^{\rig}$.  Afterwards, we can continue with the correctly
triangulated part as in Algorithm~\ref{alg:first-window}.

Concerning the shortest paths, first note that the right shortest path
$\pi_0^{\rig}$ is a suffix of the previous right shortest path, which
we already know.
For the left shortest path $\pi_0^{\lef}$, first consider the polygon
induced by the triangles that are intersected by $a'$; see
Figure~\ref{fig:second-window}c.  Let $v_1, \dots, v_g$ be the path on
the outer face of this polygon (in clockwise direction) from
$\lef(a') = v_1$ to $\lef(b_h) = v_g$.  We obtain $\pi_0^{\lef}$ using
\emph{Graham's scan}~\cite{Graham1972} on the sequence
$v_1, \dots, v_g$, i.e., starting with an empty path, we iteratively
append the next vertex of the sequence $v_1, \dots, v_g$ while
maintaining the path's outward convexity by successively removing the
second to last vertex if necessary; see
Algorithm~\ref{alg:later-window}.  We note that applying Graham's scan
to arbitrary sequences of vertices may result in self-intersecting
paths~\cite{b-chfsptd-78}.  However, we will see that this does not
happen in our case.

It remains to compute the visibility lines $\lambda_0^{\lef}$ and
$\lambda_0^{\rig}$ corresponding to the hourglass consisting of $a'$,
$b_h$, and the shortest paths $\pi_0^{\lef}$ and $\pi_0^{\rig}$.  Note
that the whole edge $b_h$ is visible from $a'$, since $a'$ intersects the
triangle $t_h$.  Thus, the visibility lines go through the endpoints
of $b_h$.  It follows that $\lambda_0^{\lef}$ is the line that goes
through $\lef(b_h)$ and the unique vertex on $\pi_0^{\rig}$ such that
it is tangent to $\pi_0^{\rig}$; see Figure~\ref{fig:second-window}b.
This can be clearly found in linear time in the length of
$\pi_0^{\rig}$.  The same holds for the right visibility line.

\begin{algorithm}[t]
  \caption{\label{alg:later-window}Computes the initial left and right
    shortest paths with corresponding visibility lines.}%
  \newcommand{\commentfont}[1]{\texttt{\color{gray}#1}}%
  \SetCommentSty{commentfont}%
  $a' = $ last window\;%
  $\pi^{\rig} = $ right shortest path computed in the previous step\;%
  $\pi_0^{\rig} = $ suffix of $\pi^{\rig}$ starting with $\rig(a')$\;%
  \tcp{Graham's scan on the sequence $v_1, \dots, v_g$}%
  $\pi_0^{\lef} = $ empty path\;%
  \For{$i = 1$ \KwTo $g$}{%
    append $v_i$ to $\pi_0^{\lef}$\;%
    \While{last bend of $\pi_0^{\lef}$ is a right bend}{%
      remove second to last element from $\pi_0^{\lef}$\;%
    }%
  }%
  $\lambda_0^{\rig} = $ tangent of $\pi_0^{\lef}$ through
  $\rig(b_{h})$\;%
  $\lambda_0^{\lef} = $ tangent of $\pi_0^{\rig}$ through
  $\lef(b_{h})$\;%
  \KwRet{$(\pi_0^{\rig}, \pi_0^{\lef}, \lambda_0^{\lef},
    \lambda_0^{\rig})$}%
\end{algorithm}

\begin{lemma}
  \label{lem:initialization-linear-time}
  Algorithm~\ref{alg:later-window} computes the initial left and right
  shortest paths $\pi_0^{\lef}$ and $\pi_0^{\rig}$ as well as the
  corresponding visibility lines in $O(|P_0'|)$ time.
\end{lemma}
\begin{proof}
  We mainly have to prove that the path $\pi_0^{\lef}$ obtained by
  applying Graham's scan on the sequence $v_1, \dots, v_g$ actually is
  the shortest path from $\lef(a)$ to $\lef(b_h)$ in $P_0'$ (which
  includes that it is not self-intersecting).  This can be seen by
  reusing arguments we made for computing the first window.  To this
  end, we reuse the triangulation we have for $P$ by placing new
  vertices where $a'$ crosses with triangulation edges; see
  Figure~\ref{fig:second-window}d.  Note that the resulting polygon,
  which we denote by $P_0''$, is almost triangulated, i.e., each face
  is a triangle or a quadrangle.  We can thus triangulate $P_0''$ by
  adding one new edge in each quadrangle as in
  Figure~\ref{fig:second-window}d.  Note that $a'$ is separated into
  several edges in $P_0''$; let $a''$ be the topmost of these edges
  (i.e., the last one in clockwise order).  Assume, we want to compute
  the minimum-link path from $a''$ to $b_h$ in $P_0''$.  First note
  that the triangle $t_h$ is clearly visible from $a''$.  Thus, our
  algorithm for computing the first window computes the shortest
  path from $\lef(a') = \lef(a'')$ to $\lef(b_h)$.  Note further that
  the vertices visited in
  Lines~\ref{line:extend-path-start}--\ref{line:extend-path-end} of
  Algorithm~\ref{alg:first-window} are the vertices $v_1, \dots, v_g$
  in this order.  Thus, Algorithm~\ref{alg:first-window} actually
  constructs the left shortest path by using Graham's scan on the
  sequence $v_1, \dots, v_g$.  It follows that directly applying
  Graham's scan to the sequence $v_1, \dots, v_g$ correctly computes
  the left shortest path in $P_0'$ from $\lef(a')$ to $\lef(b_h)$.
  Clearly, the running time of Algorithm~\ref{alg:later-window} is
  linear in the size of $P_0'$.
\end{proof}

We compute subsequent windows as shown before, until the last edge~$b$ is found. The actual minimum-link path~$\apath$ is obtained by connecting each window~$w(a)$ to its corresponding first edge~$a$ with a straight line~\cite{Sur86}.
Linear running time of the algorithm follows immediately from Lemma~\ref{lem:window-linear-time} and Lemma~\ref{lem:initialization-linear-time}.

\paragraph{Implementation Details.}

\begin{figure}[t]
\centering
\includegraphics[page=1]{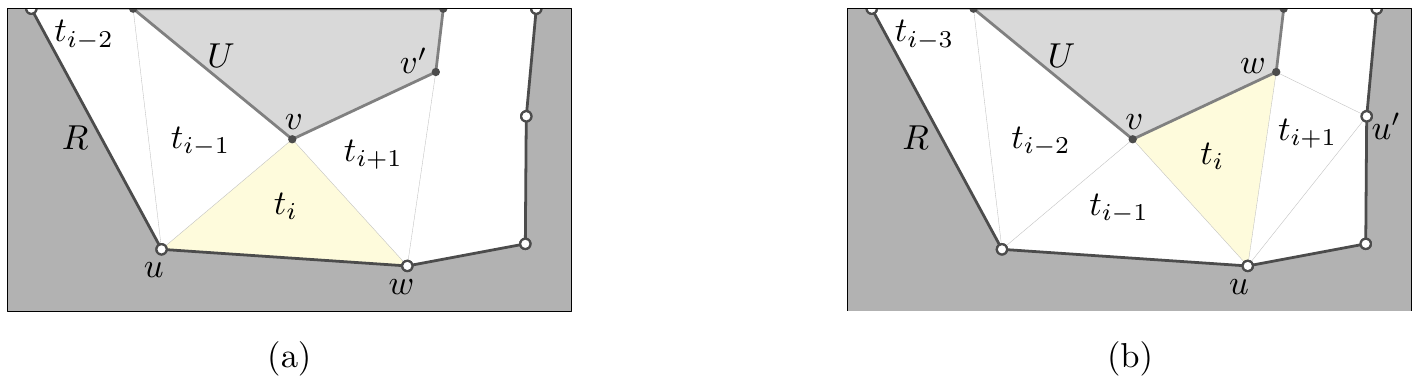}
\caption{(a)~The next important triangle $t_{i+1}$ shares with~$t_i$ the unique edge $vw$ that is not~$uv$ and has exactly one reachable endpoint. (b)~The next important triangle $t_{i+1}$ contains the edge~$uw$.}
\label{fig:important-triangle}
\end{figure}

To obtain the desired polygon that separates~$R$ and~$U$, we can connect the first and last segment of~$\apath$ along the boundary edge~$e$, as described above. However, to potentially save a segment and for aesthetic reasons, we first test whether the last window can be extended to the first segment of the path without intersecting the boundary of~$R$ or~$U$. This can be done by continuing the computation of the last window from~$t_0$ after~$b$ was found.
Moreover, we do not construct~$P$ and its triangulation~$T$ explicitly, but work directly on the triangulated input graph. The next important triangle is then computed on-the-fly as follows. Consider an important triangle~$t_i = uvw$, and
let $uv$ be the edge shared by the current and the previous important triangle; see Figure~\ref{fig:important-triangle}. Clearly, exactly one endpoint
of~$uv$ is part of the reachable boundary, so without loss of generality let~$u$ be this
endpoint. Then the next important triangle is the triangle
sharing~$vw$ with~$t_i$ if~$w$ is reachable, and the triangle
sharing~$uw$ with~$t_i$ otherwise. In other words, the next
triangle is determined by the unique edge that has exactly one
reachable endpoint.
Triangles are stored in a single array, similar to the data structure to store faces of the planar graph described in Section~\ref{sec:range_query} (although we do not need sentinels, since triangular faces have constant size). The minimum-link path algorithm operates on this data structure to obtain the sequence of important triangles on-the-fly.

\section{Heuristic Approaches for General Border Regions}
\label{sec:heuristics}

A border region~$B = R \cup U$ may consist of several unreachable components, \ie,~$|U| > 1$, while $|R| = 1$ always holds. In the general case, it is not clear whether one can compute a (non-intersecting) range polygon of minimum complexity that separates $R$ and $U$ in polynomial time~\cite{Gui93}. Therefore, we propose four heuristic approaches with (almost) linear running time (in the size of~$B$). Figure~\ref{fig:heuristics} shows example outputs of the different heuristics. Their quality is evaluated on realistic input instance in Section~\ref{sec:experiments}.

\begin{figure}[tb!]
  \centering
  \includegraphics{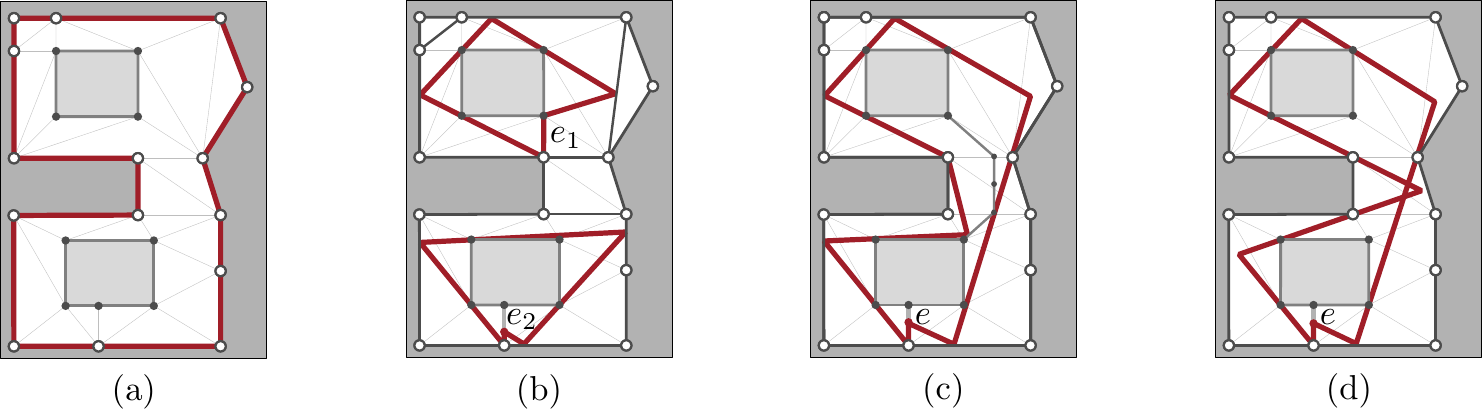}
  \caption{Example output (red) of the different approaches for a border region with two unreachable components. Minimum-link paths are computed from the bounding edge~$e$ in (b) and (d), and from $e_1$ and $e_2$ in (c). }
  \label{fig:heuristics}
\end{figure}

Given a border region~$B$, the approach presented in
Section~\ref{sec:heuristics:reachable_component_extraction} simply returns the reachable boundary~$R$ (Figure~\ref{fig:case_study_polygon_with_holes}a, Figure~\ref{fig:heuristics}a). Results are similar to previous algorithms for isochrones~\cite{Mar10}. Since the
complexity of the range polygons can become quite high, we
propose more sophisticated heuristics.
The basic idea of the approach
introduced in Section~\ref{sec:heuristics:triangle_restricted} is to
use the graph triangulation to separate~$B$ along edges for which both endpoints are in~$R$. The modified instances consist of single unreachable components that are separated from the reachable component by the algorithm in Section~\ref{sec:min_link_path} (Figure~\ref{fig:heuristics}b).
In
Section~\ref{sec:heuristics:connected_components}, we propose to insert new edges that connect the components of $U$ to create an instance with~$|U| = 1$. Then, we compute a minimum-link path in
the resulting border region as in
Section~\ref{sec:min_link_path} (Figure~\ref{fig:case_study_polygon_with_holes}b, Figure~\ref{fig:heuristics}c).
Finally,
Section~\ref{sec:heuristics:self_intersecting} details a heuristic that modifies the approach of Section~\ref{sec:min_link_path} to compute a possibly self-intersecting minimum-link path separating $R$ and~$U$ with minimum number of segments; see Figure~\ref{fig:heuristics}d.
Consequently, the resulting polygon has at most two more segments than an optimal solution.
We rearrange this polygon at intersections to obtain a range polygon without self-intersections.

\subsection{Extracting the Reachable Component}
\label{sec:heuristics:reachable_component_extraction}

Given a border region~$B$, this approach returns the reachable boundary~$R$. The resulting range polygon closely resembles the
results of known approaches, which essentially consist of extracting
the reachable subgraph~\cite{Gam12,Mar10}.
Note that this approach does not have to compute the unreachable boundary explicitly. We can improve its performance by modifiying the extraction of border regions described in Section~\ref{sec:range_query}, such that only the reachable part of the boundary is traversed.
In a sense, our modified extraction algorithm can be
seen as an efficient implementation of these previoues approaches~\cite{Gam12,Mar10}.
Its linear running time (in the size of~$B$) follows from the fact that we traverse every edge of~$R$ once, and every boundary edge or accessible edge of the planar graph~$\graph_\planar$ contained in~$B$ at most twice (the initial edge is visited twice, all other edges once). Clearly, the number of these edges is linear in the size of~$B$.

\subsection{Separating Border Regions Along their Triangulation}
\label{sec:heuristics:triangle_restricted}

The idea of this approach is as follows. For each border region~$B$, we consider its triangulation. We add all edges of the triangulation that either connect two reachable vertices or two unreachable vertices of~$\graph_\planar$ to~$B$, possibly splitting~$B$ into several regions~$B' = R' \cup U'$ (see thick edges separating the border region in Figure~\ref{fig:heuristics}b). For each region~$B'$, we obtain~$|U'| \le 1$, since two components of~$U$ must be connected by an edge of the triangulation or separated by an edge with two endpoints in~$R$. Then, we run the algorithm presented in Section~\ref{sec:min_link_path} on each instance~$B'$ with~$|U'| = 1$ to get the range polygon.
Linear running time follows, as we run the linear-time algorithm of Section~\ref{sec:min_link_path} on disjoint subregions of~$B$.

Clearly, the number of edges we add to~$B$ is not minimal, \ie, in general we could omit some of them and still obtain $|U'| \le 1$ for each region~$B'$. On the other hand, computing the set of separating edges described above is trivial, making our approach very simple.
In what follows, we describe how it is implemented without explicitly computing the border region~$B$. Instead, we use the set~$\edges_x$ to identify
border regions that need to be handled (recall that~$\edges_x$ contains all boundary edges and accessible edges of the border regions; see Section~\ref{sec:range_query}). We loop over all edges in this
set and check for each boundary edge~$(u,v) \in \edges_x$
whether it was already visited. If this is not the case, we start a
minimum-link path computation from the triangle left of this edge
(we ignore accessible edges, since they no longer
intersect the interior of the modified border regions).
The sequence of important triangles is
then computed on-the-fly as described in Section~\ref{sec:min_link_path}.
%
%
%
Whenever the algorithm passes a boundary edge in~$\edges_x$, it is marked as visited.

This heuristic can be seen as a simple but effective way of
producing border regions with a single unreachable
component. It is very easy to implement and even simplifies
aspects of the minimum-link path algorithm described in
Section~\ref{sec:min_link_path}, because the modified border regions
contain only important triangles (all other triangles were removed
from the modified border region; see Figure~\ref{fig:heuristics}b). Thus, computations for finding the second endpoint of a new
window and the initial visibility lines
are restricted to a single triangle, enabling
the use of simpler data structures.
On the other hand, the result of the algorithm heavily depends on the
triangulation of the input graph. In addition to that, the number of
modified regions~$B'$ can become quite large (see Section~\ref{sec:experiments}). The approach presented
in the next section therefore proposes a more sophisticated way to
obtain regions with a single unreachable component.

\subsection{Connecting Unreachable Components}
\label{sec:heuristics:connected_components}

This approach adds new edges to border regions with more than one unreachable component, such that they connect all unreachable components without intersecting the reachable boundary; see Figure~\ref{fig:heuristics}c. We obtain a modified instance~$B'$ with a single unreachable component and apply the algorithm from Section~\ref{sec:min_link_path}. Note that unreachable components can not always be connected by straight lines; see Figure~\ref{fig:connected_components}b.
For a similar (more general) scenario, Guibas et al.~\cite{Gui93} propose an approach to computes a subdivision that requires $\bigO(h)$ more segments than an optimal solution (where $h$ is the number of components in the input region).
Here, we propose a heuristic approach without any nontrivial
worst-case guarantee. However, it is easy to implement
and provides high-quality solutions in practice (see
Section~\ref{sec:experiments}).

\begin{figure}[t]
  \centering
  \includegraphics{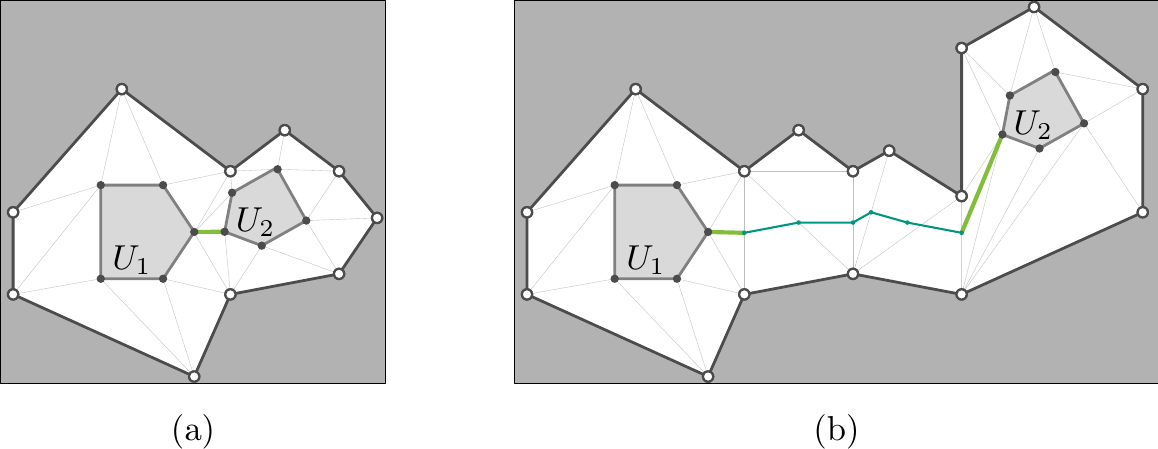}
  \caption{A border region $B$ with two unreachable components $U_1$ and $U_2$. (a)~Components $U_1$ and $U_2$ are connected by an edge (green) of the triangulation. (b)~The unreachable components $U_1$ and $U_2$ are connected by two connecting edges (green, thick) and several bridging edges (green, thin).}
  \label{fig:connected_components}
\end{figure}

Given a border region~$B$ whose unreachable boundaries~$U$ consist of several components, we describe our algorithm that computes a modified region~$B'$ with a single unreachable component.
It runs a breadth first search (BFS) on the dual
graph of the triangulation of~$B$ to find paths that connect the unreachable
components. Then, we add new edges that connect the components and
retriangulate the modified border region.

We distinguish two cases for connecting two unreachable
components~$U_1$ and $U_2$ in~$U$.
First,~$U_1$ and~$U_2$ may be
connected by a single edge of the triangulation (\ie, a edge that has
an endpoint in each component). Then, we can add this edge to~$B$ to
connect~$U_1$ and~$U_2$; see Figure~\ref{fig:connected_components}a.
Second, we have to deal with components that
are not connected by such an edge. In this case, there exists at least one
edge~$e$ in the triangulation of~$B$, such that both endpoints of~$e$
are on the reachable boundary, and $e$ separates $B$ into two
subregions containing~$U_1$ and~$U_2$, respectively. Hence, any
path connecting~$U_1$ and~$U_2$ in~$B$ crosses~$e$.
%
%
Our goal is to find a short path in the dual graph of the triangulation that
connects~$U_1$ and~$U_2$. Then, we add new edges to the corresponding
sequence of triangles to connect~$U_1$ and~$U_2$ in~$B$; see Figure~\ref{fig:connected_components}b. Afterwards,
we locally retriangulate the modified part of~$B$.

\paragraph{Connecting Components.}

Our algorithm starts by checking for each pair of components, whether they are connected by a single edge in the triangulation of~$B$. This can be done in a sweep over the vertices of each unreachable component, scanning for each vertex its outgoing edges in the triangulation. Whenever an edge is found that connects two unreachable components, we merge these components and consider them as the same component in the further course of the algorithm (making use of a union-find data structure).

To connect all remaining unreachable components after this first step, we proceed as follows.
%
%
Consider the (weak) dual graph of the triangulation of~$B$.
Since no pair of the remaining unreachable components can be connected by a single edge in the primal graph, each triangle intersects at most one unreachable component. We assign a component to each dual vertex, namely, the reachable component if the corresponding triangle contains only reachable vertices, otherwise the unique unreachable component this triangle intersects.
For each unreachable component, we add a super source vertex to the dual graph that is connected to all vertices assigned to this component.
Now, our goal is to find a tree of minimum total length in this graph that connects all super sources, \ie, a minimum Steiner tree. Since this poses an \NP-hard problem in
general~\cite{Garey:1979:CIG:578533}, we propose a heuristic
search.
The basic idea is to iteratively add shortest paths between two sources that are not connected yet in a greedy fashion. This can by achieved by a multi-source variant of a BFS, which we now describe in detail.

Given an undirected graph~$\graph$ and a set of~$k$ source vertices, it keeps a list of vertex labels~$\ell(\cdot)$, to mark visited vertices, store their parents in the search, and the source vertex of the path that reached this vertex.
Moreover, we use a union-find data structure (with~$k$ elements) to maintain connectivity of sources.
Initially, we mark all source vertices as visited and set each as its own source. Moreover, all source vertices are inserted into a queue.
Then, the main loop is run until all source vertices are connected. In each step, the search extracts the next vertex~$u$ from the queue and checks all incident edges~$\{u,v\}$. If~$v$ was not visited, it marks its label as visited, and sets its source to the source vertex of the label of~$u$. Additionally,~$v$ is inserted into the queue. On the other hand, if~$v$ was visited, we check whether the sources of the labels $\ell(u)$ and~$\ell(v)$ are connected. If they are not, we found a path that connects both sources. Thus, we unify the sources (i.e., they are considered equal in the further course).
The actual path can be retrieved by backtracking from~$u$ and~$v$, respectively, following the parent pointers until the source is reached. The concatenation of both paths yields a path that connects two source vertices.
The algorithm stops when all sources are connected.

%
%
After the search terminates, we \emph{split} some triangles by adding new vertices and edges to~$B$ in the following manner; see Figure~\ref{fig:connected_components}b.
Consider a path in the dual graph connecting two components~$U_1$ and~$U_2$. First, we remove the first and last vertex of this path (since these are previously added super sources). For the remaining path, consider
the corresponding sequence of triangles in~$B$. Clearly, all but the
first and last triangle of this path only have endpoints in the
reachable component (otherwise, backtracking would have started or
stopped earlier).
For every edge shared by two triangles in this path, we add a \emph{bridging vertex} at the center of this edge.
Thus, a bridging vertex is always
contained in an edge with two reachable endpoints. 
Between any pair of bridging vertices contained in the same triangle, we add a \emph{bridging edge} connecting them.
Finally, at the first and
last triangle of the path, we add a \emph{connecting edge} from the unique endpoint that
belongs to an unreachable component to the bridging vertex.
Assigning all added vertices and edges to the unreachable boundary, the resulting border region~$B'$ contains a connected unreachable component $U' \supseteq U_1 \cup U_2$.

For correctness, we need to show that the union of all edges added according to the computed Steiner tree creates no crossings. First, note that bridging edges in a triangle (corresponding to different subpaths of the Steiner tree) never cross each other.
Second, we claim that if a connecting edge is inserted in some triangle, no other edge is added to that triangle.
Since a connecting edge has an endpoint in some unreachable component, the triangle contains at most one edge with two reachable endpoints. Hence, it contains at most one bridging vertex, and therefore no bridging edge. Moreover, the triangle contains at least one bridging vertex, which is the other endpoint of the connecting edge. Thus, it has exactly two reachable endpoints and does not contain more than one connecting edge.

Finally, we add new edges (if necessary) to any created quadrangles to maintain the triangulation. Thus, the resulting modified border region~$B'$ is triangular and its unreachable boundary consists of a single component.
We run the algorithm described in Section~\ref{sec:min_link_path} to obtain the desired range polygon.

The BFS described above visits each vertex of the dual graph at most once. In each step, the union-find data structure is called at most three times, to check whether sources of two given labels are connected and unify them if necessary. All other operations require constant time.
Using path compression for the union-find data structure~\cite{Tar75}, this yields a running time of~$\bigO(n \alpha(n))$ of the BFS, where $n$ is the number of vertices in the dual graph (which is linear in the size of~$B$) and $\alpha$ the inverse Ackermann function.
Since the remaining steps of the heuristic (adding vertices and edges to triangles, computing a minimum-link path) require linear time, the overall running time is almost linear.

\paragraph{Improvements.} In practice, the performance of the BFS is dominated by the number of visited vertices.
We propose tuning options that reduce this number significantly, without affecting correctness of the approach (although results may slightly change).

One crucial observation is that realistic instances of border
regions~$B$ often have an unreachable boundary consisting of one large component
(the major part of the unreachable subgraph), and
many tiny components (\eg, unreachable dead ends in
the road network), similar to Figure~\ref{fig:border_region}. Then, the search from the large component dominates the running time. Instead, we can run the
BFS starting from all but the largest component. This requires only
negligible overhead (we identify the largest component in an
additional sweep), but searches from small components are likely to
quickly converge to the large component. In preliminary experiments,
this reduced running time significantly.
Furthermore, after extracting the next vertex from the queue, we first check whether its source was connected to the largest component in the meantime. If this is the case, we prune the search at this vertex (because it now represents the search from the largest component).
Similarly, we omit sweeping over vertices of the largest component when checking for edges in the triangulation that connect two components (before running the BFS).

Going even further, we always expand the search from the component that is currently the smallest.
In its basic variant, the BFS uses a queue to process vertices in first-in-first-out order.
For better (practical) performance, we replace it by a priority queue whose elements are components (represented by source vertices) instead of vertices. Additionally, we maintain a queue for each component that stores the vertices (extracting them in first-in-first-out order).
In the priority queue, each component uses its complexity (\ie, its number of edges in the border region) as key.
In each step of the BFS, we check for the component with the smallest key in the priority queue, and extract the next vertex from the queue of this component. If it has run empty, we remove the component from the priority queue.
New vertices are always added to the queue that corresponds to their source
label.
%
%
Keys in the priority queue are updated accordingly whenever components
are unified (i.e., we update the key of all affected components -- this requires linear
time in the number of contained components).
Note that the use of a priority queue together with this simple update routine increase the asymptotic running time of the BFS by a linear factor (in the number of components, which can be linear in the size of the border region). However, we observe a significant speedup in practice.

\paragraph{Data Structures and Implementation Details.}

When running a BFS on the dual graph of a border region~$B$, we implicitly represent the search graph using the triangulation of the graph~$\graph_\planar$. To determine incident edges of a dual vertex, we check the edges of its primal triangle. If the primal edge is not contained in~$\graph_\planar$ (\ie, it was added during triangulation), or if it is contained in~$\edges_x$, there exists an edge in the dual graph connecting the triangle to the twin triangle of this primal edge.
%
%
%
%
%
We maintain flags at each triangle set during backtracking, to determine whether a bridging edge or a connecting edge should be inserted (and if so, between which pair of endpoints).
Backtracking is run on-the-fly during the
BFS, and stopped whenever we reach a previously split triangle (since this means
we have reached a previously computed path).
We build a list of all split triangles, for fast (sequential) access to all
triangles that were split after the BFS has terminated, in order to add the respective edges to the triangulation.

To avoid costly reinitialization of the vertex labels between queries, we make use of timestamps, implicitly encoded within the component indices to save space. After each query, the global timestamp is increased by the number of unreachable components of~$B$. When storing a component index in a label, it is increased by the global timestamp. Then, a label is invalid if this index is below the global timestamp. To retrieve the actual index of a valid label, we subtract the global timestamp.


Edges and vertices added to the border region~$B'$ are stored as temporary modifications in the triangulation of~$\graph_\planar$. To this end, we make use of the following data structures. Edges of the triangulation that are added to~$U'$ (to connect two unreachable components) are explicitly stored in a list. To quickly check whether some edge of the triangulation was added to~$U'$, we sort this list (\eg, by head vertex index) after the BFS terminated to enable binary search. In our setting (some 1\,000 inserted edges for the hardest queries), this was slightly faster than using hash sets.
To store the bridging edges and connecting edges, we temporarily modify the
triangulation. To this end, we add an invalidation flag and a
temporary index to the vertices of every triangle in the triangulation of~$\graph_\planar$.
Moreover, we maintain a list of temporary triangles. Each vertex
in a split triangle is marked as invalid and its temporary index is set to the corresponding entry in this list. Before retrieving a triangle vertex, we first check whether it is invalid and redirect to the temporary vertex if this is the case.
For faster reinitialization, we replace invalidation flags by timestamps, similar to component timestamps described above.
%
%
When splitting triangles, we have to set the twins of all new edges.
If the twin triangle was not created yet, we
store the pending edge in a list. This list is searched for existing
twins whenever a new triangle is added. If a twin is found, we set
twins for both affected edges, and remove them from the set.
%



\subsection{Computing Self-Intersecting Minimum-Link Paths}
\label{sec:heuristics:self_intersecting}

Our last approach computes a minimum-link path in~$B$ that separates the reachable boundary from the unreachable boundaries.
While the resulting polygon has at most~$\opt+2$ segments, it may intersect itself; see Figure~\ref{fig:heuristics}d.
To obtain a range polygon from a self-intersecting polygon, we rearrange it accordingly at intersections.

The remainder of this section focuses on computing minimum-link paths in border regions with several unreachable components. To achieve this, we have to make some modifications to the algorithm described in Section~\ref{sec:min_link_path}.
First, note that the (weak) dual graph of the triangulation of~$B$ is not outerplanar if~$|U| > 1$. Consequently, paths between (dual) vertices are no longer unique.
Instead, we require \emph{shortest} paths in the dual graph that separate reachable and unreachable boundaries.
In fact, some vertices may occur several times in such paths; see the corresponding sequence of triangles crossed by the polygon in Figure~\ref{fig:heuristics}d.
In what follows, we first show how to obtain the sequence of important triangles in this general case.
Then, we describe modifications necessary to retain correctness of the algorithm described in Section~\ref{sec:min_link_path} running on this sequence of important triangles.

\paragraph{Computing the Important Triangles.}

Given a boundary edge~$e$ of the border region~$B$, we are interested in a minimum-link path that connects both sides of~$e$ and separates the reachable boundary from all unreachable boundaries.
We compute a sequence~$t_1, \dots, t_k$ of triangles such that any minimum-link path with the above property must pass this sequence in this order.
Our approach runs in two phases. The first phase traverses the reachable boundary of~$B$ and lists all encountered boundary edges in the triangulation (\ie, all edges with one endpoint in each~$R$ and~$U$, even if they are not present in the input graph~$\graph_\planar$). Clearly, the minimum-link path must intersect all boundary edges in the same order (lest having unreachable components on both sides of the path). The second phase uses this information to compute the actual sequence of important triangles, consisting of shortest paths in the dual between pairs of consecutive boundary edges. It serves as input for the algorithm that computes a minimum-link path connecting both sides of~$e$ in~$B$.

During the first phase, we exploit the fact that the reachable boundary of~$B$ is always connected. We assign \emph{indices} to all edges in the triangulation of~$B$ contained in the border region that intersect the reachable boundary, according to the order in which they are traversed in each direction starting from~$e$; see Figure~\ref{fig:computing-important-triangle}.
For consistency, sides of edges that are not traversed get the index~$\infty$.
Clearly, this information can be retrieved in a single traversal of the reachable boundary, similar to the procedure described in Section~\ref{sec:heuristics:reachable_component_extraction}, but running on the triangulation of the border region.
%
%
%
Moreover, during this traversal we collect an ordered list of indices corresponding to boundary edges. Observe that every boundary edge in~$B$ is traversed exactly once.
%

The second phase runs on the dual graph of the triangulation and retrieves the desired sequence of triangles.
A key observation is that this sequence must pass all boundary edges exactly once and in increasing order of their indices.
Therefore, we can compute the sequence of important triangles as follows.
\begin{figure}
\centering
\includegraphics[page=1]{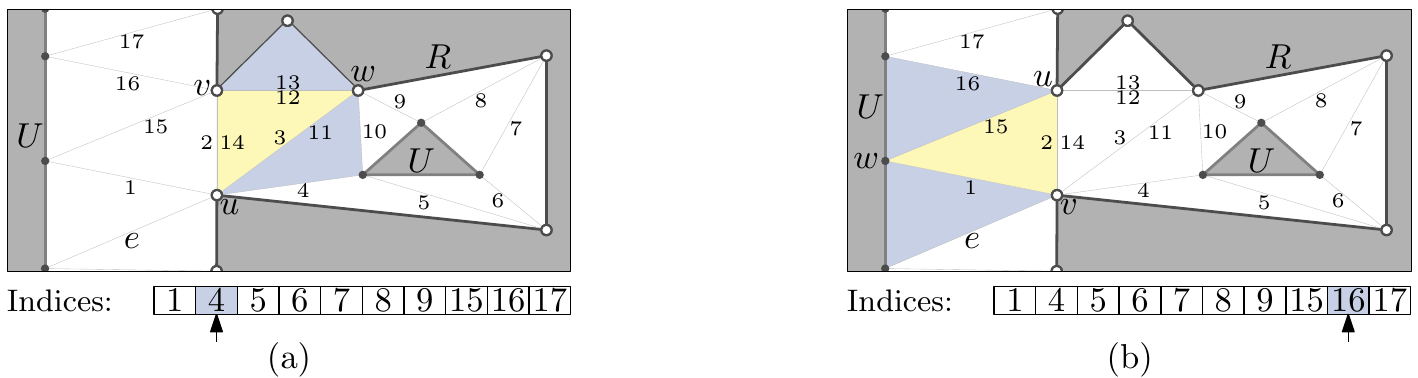}
\caption{Border region showing indices of edges when starting traversal at~$e$. Note that edges have two indices (one for each direction of traversal). Indices are~$\infty$ if not specified. The list of boundary edges is given in the array, indicating the next index with an arrow. (a) The third visited triangle in the second phase (shaded yellow) has two possible next triangles $t_{uw}$ and~$t_{vw}$ (shaded blue). The next triangle is $t_{uw}$, because the index of the edge~$vw$ with greater index (12) exceeds the next boundary edge index (4).
%
%
(b) The next triangle is $t_{uw}$. The index of the next boundary edge is updated to 16.}
\label{fig:computing-important-triangle}
\end{figure}
We maintain the index of the next boundary edge that was not traversed yet, initialized to the first element of the list. Starting at the triangle~$t_1$ containing the first boundary edge~$e$, we add triangles to the sequence of important triangles until~$e$ is reached again. Let~$t_i = uvw$ denote the previous triangle that was appended to this sequence. Then we determine the next triangle~$t_{i+1}$ as follows; see Figure~\ref{fig:computing-important-triangle}.
Let~$uv$ be the unique edge shared by~$t_i$ and~$t_{i-1}$ (in the case of~$i = 1$, we have $uv = e$).
%
%
%
%
To determine the next triangle, we consider the two possible triangles $t_{uw}$ containing the edge $uw$ and $t_{vw}$ containing $vw$.
Without loss of generality, let the index of $uw$ be smaller than the index of~$vw$ (and thus, finite). This implies that $uw$ is not contained in the boundary of~$B$ (otherwise, it would have index~$\infty$).
If both $u$ and $w$ are part of the reachable boundary, we know that $uw$ separates~$B$ into two subregions; see Figure~\ref{fig:computing-important-triangle}a. Thus, $t_{uw}$ is the next triangle if and only if the subregion~$B_u$ containing~$t_{uw}$ contains a boundary edge
%
%
that was not passed yet.
Therefore, we continue with $t_{uw}$ if and only if the index of the other edge~$vw$ is greater than of the next boundary edge.
If either $u$ or $w$ is part of an unreachable boundary, $uw$ is the next boundary edge; see Figure~\ref{fig:computing-important-triangle}b.
%
%
We update the index of the next boundary edge to the next element in the according list.


We continue until the first edge~$e$ is reached again. The resulting sequence of triangles is the desired shortest path.
Note that the second phase (traversing the dual graph) can be performed on-the-fly during minimum-link path computation (\ie, the sequence of triangles does not have to be built explicitly).

\paragraph{Computing Minimum-Link Paths.}

Given a sequence~$t_1, \dots, t_k$ of important triangles in a border region~$B$ computed as described above, we discuss how to compute a minimum-link path between both sides of~$e$. In particular, we show which modifications to the approach presented in Section~\ref{sec:min_link_path} are necessary at certain points to preserve correctness.

Consider the computation of a window from an arbitrary \emph{initial edge}~$b$ shared by two triangles in~$t_1, \dots, t_k$ as described in Section~\ref{sec:min_link_path}. Clearly, the subsequence of triangles that is visited until a window is found does not contain several occurrences of the same triangle, since this would imply that a straight visibility line intersects it at least twice. Consequently, the fact that triangles may appear several times in~$t_1, \dots, t_k$ does not affect window computation starting at an edge.
A similar argument applies when initializing the computation of a subsequent window. Recall that in this step, the visibility cone from the last window to the next initial edge is computed.
All triangles considered in this step are intersected by the last window in a certain order (see Section~\ref{sec:min_link_path}), so we do not encounter the same triangle twice (since windows are straight lines).

\begin{figure}[t!]
  \centering
  \includegraphics{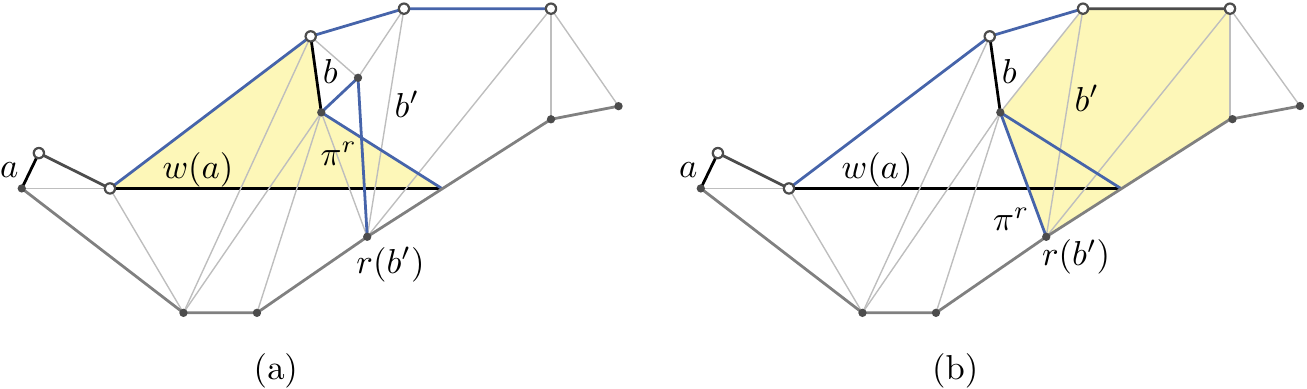}
  \caption{(a)~The shortest path~$\pi^r$ from the right endpoint of the previous window~$w(a)$ to the right endpoint of~$b'$ intersects itself. Note that there are two unreachable vertices that are not connected to the remaining unreachable boundary. (b)~The shortest path~$\pi_r$ from the right endpoint of~$w(a)$ to~$r(b')$ contains a left bend after passing an unreachable vertex that is not connected to the unreachable boundary.}
  \label{fig:self_intersecting_outward_path}
\end{figure}

However, the computation of the next window after this initialization step requires some modification, since triangles visited during initialization may reoccur when computing the window from the next edge.
As a result, the subpaths computed during initialization and when starting from this edge may intersect each other; see Figure~\ref{fig:self_intersecting_outward_path}a.
In this example, the subpath of the right shortest path~$\pi^{\rig}$ starting at the initial edge~$b$ intersects the segment from the right endpoint of~$w(a)$ to the right endpoint of~$b$.
Self-intersections would not pose a problem per se if we would generalize the definition of shortest paths to polygons with self-intersections. However, without modifications, Algorithm~\ref{alg:first-window} in Section~\ref{sec:min_link_path} may produce wrong results in certain special cases. Figure~\ref{fig:self_intersecting_outward_path}b shows such an example. While the shortest path from the right endpoint of~$w(a)$ to the right endpoint of~$b'$ does not intersect itself in this case, its second segment lies in the half plane to the left of the first segment. Hence, Algorithm~\ref{alg:first-window} will falsely remove the last bend. The resulting incorrect path consists of the single segment from the right endpoint of~$w(a)$ to~$r(b')$.
Clearly, this leads to the construction of an incorrect visibility cone.
We say that the last segments of the right shortest path shown in Figure~\ref{fig:self_intersecting_outward_path} are \emph{visibility-intersecting}, as they reach into the area that is visible from~$w(a)$.
Formally, a segment is visibility-intersecting if it intersects the interior of the hourglass~$H_0$ bounded by the previous window~$w(a) = a'$, the next initial edge~$b$, and the initial left and right shortest path~$\pi_{0}^{\lef}$ and $\pi_{0}^{\rig}$; see the shaded area in Figure~\ref{fig:self_intersecting_outward_path}a.
Note that visibility-intersecting segments can only occur in the shortest path that corresponds to the unreachable boundary (since the reachable boundary consists of a single component).

In what follows, we show how we can avoid visibility-intersecting segments that may spoil the algorithm presented in Section~\ref{sec:min_link_path}.
Clearly, visibility-intersecting segments only occur if a triangle that was visited during the initialization phase is visited again when computing the next window from a boundary edge.
We could resolve this issue conceptionally easy by retriangulating parts of the border region (namely, the part called~$P'_0$ in Section~\ref{sec:min_link_path}). Below, we present an approach that avoids retriangulation by making use of few simple checks instead. (In a sense, it simulates the situation after such a retriangulation.)
To this end, we show how we can easily detect visibility-intersecting segments. Then, we show that we can simply omit such segments from the corresponding shortest path.

We claim that a segment that is appended to the shortest path~$\pi^{\rig}$ is visibility-intersecting if and only if this segment intersects the previous window (and is not an endpoint of this window).
For the sake of simplicity, we assume general position. Thus, the window~$a$ and the path~$\pi^{\rig}$ share no common segment. In practice, such a segment can easily be detected and removed from both the path and the window during initialization of~$\pi_0^{\rig}$.

\begin{lemma}
Given the previous window~$w(a)$, let~$b$ be the next initial edge of $B$. Let~$t_i$ ($i \ge 1$) be an important triangle such that the shortest path~$\pi_{i-1}^{\rig}$ contains no visibility-intersecting segments and the edge~$b_i$ shared by~$t_i$ and~$t_{i+1}$ is (partially) visible from~$w(a)$.
The next segment~$s$ appended to~$\pi_{i-1}^{\rig}$ (\ie, the unique edge shared by~$t_i$ and the boundary of~$B$) is visibility-intersecting if and only if~$s$ intersects the open line segment~$w(a)$.
\end{lemma}

\begin{proof}
Note that we consider the open line segment~$w(a)$, since its right endpoint coincides with an endpoint of the first segment of~$\pi_{i-1}^{\rig}$, which clearly is not visibility-intersecting.

First, assume the segment~$s = pq$ is visibility-intersecting and assume for a contradiction that~$s$ does not intersect~$w(a)$.
Since $s$ is visibility-intersecting, it intersects the interior of the hourglass~$H_0$ enclosed by~$w(a)$, $b$, $\pi_0^{\rig}$ and~$\pi_0^{\lef}$.
Since the interior of~$H_0$ contains no vertices (in particular, neither~$p$ nor~$q$), the edge~$s$ of~$t_i$ must intersect the boundary of~$H_0$ at least twice.
However, $s$ does not intersect the interior of the edge~$b$, since both~$s$ and~$b$ are edges of triangles.
As~$s$ does not intersect~$w(a)$ by assumption, it intersects $\pi_0^{\rig}$ or~$\pi_0^{\lef}$.
Moreover, since both paths are concave in~$H_0$, $s$ must intersect both paths.
(If it intersects any path twice at two points~$p'$ and~$q'$, the subsegment that connects~$p'$ and~$q'$ does not intersect the interior of~$H_0$, so $s$ has at least one additional intersection with the boundary of~$H_0$.)
%
%
But $s$ does not intersect~$\pi_0^{\lef}$, because this would imply that the paths~$\pi_i^{\rig}$ and~$\pi_0^{\lef}$ have non-empty intersection, contradicting the fact that~$b_i$ is visible from~$w(a)$.
%

Second, assume that~$s$ intersects the open segment~$w(a)$. Since~$w(a)$ contains no endpoint of~$s$, we know that~$s$ intersects the interior of~$H_0$. Hence,~$s$ is visibility-intersecting.
\end{proof}

Next, we show that visibility-intersecting segments can safely be omitted from the shortest path computed by the algorithm.
Let $t_i, \dots, t_j$ be a subsequence of important triangles, such that the edge of~$t_i$ appended to~$\pi_{i-1}^{\rig}$ by the algorithm of Section~\ref{sec:min_link_path} is visibility-intersecting, and~$t_j$ is the first triangle (\ie, with lowest index~$j > i$) such that the edge~$b_j$ shared by~$t_j$ and~$t_{j+1}$ does not intersect the open segment~$a$; see the sequence of shaded triangles in Figure~\ref{fig:self_intersecting_outward_path}b.
Thus, all edges~$b_i, \dots, b_{j-1}$ intersect~$a$. Moreover, $\pi_{i-1}^{\rig}$ and~$b_j$ lie on the same side of~$a$ (otherwise, the window~$a$ would cross the reachable boundary of~$B$).
We distinguish two cases, depending on whether~$b_j$ is visible from~$a$. We show that in both cases, we can skip all right endpoints of the edges~$b_i, \dots, b_{j-1}$ when updating the path~$\pi^{\rig}$ to obtain the correct window.

First, assume that $b_j$ is (partially) visible from~$a$. We claim that no right endpoint of an edge~$b_i, \dots, b_{j-1}$ is contained in the shortest path~$\pi_{j}^{\rig}$. To see this, let~$u$ denote the last vertex of~$\pi_{i-1}^{\rig}$ and~$w$ the right endpoint of~$b_j$.
%
%
Clearly, $a$ separates $u$ and $w$ from all right endpoints of the edges~$b_i, \dots, b_{j-1}$. Since~$b_j$ is visible, this implies that the segment~$uw$ crosses all edges~$b_i, \dots, b_{j-1}$. Therefore, it does not intersect the boundary of~$B$ and there can be no shorter path from $u$ to $w$ passing any right endpoint of these edges.

Second, assume that $b_j$ is not visible from~$a$. We claim that no right endpoint~$v$ of an edge~$b_i, \dots, b_{j-1}$ is contained in the visibility cone from~$a$.
Assume for a contradiction that such an endpoint~$v$ is visible from~$a$. We know that~$v$ and the edge~$b_{i-1}$ lie on opposite sides of~$a$. Since $v$ is visible from~$a$, there exists a straight line that crosses $a$, $b_{i-1}$ and~$v$ in this order. Consequently, it must cross~$a$ twice, contradicting the fact that it is a straight line. Since~$v$ is not part of the visibility cone from~$a$, it is not relevant for the computation of the next window~$w(a)$.

In both cases, we can safely ignore right endpoints of all edges~$b_i, \dots, b_{j-1}$.
We adapt our algorithm as follows.
Before adding a segment to the shortest path that corresponds to the unreachable boundary, we check whether it intersects the previous window~$a$. If this is the case, we do not add it to the path.

\begin{figure}[t]
  \centering
  \includegraphics{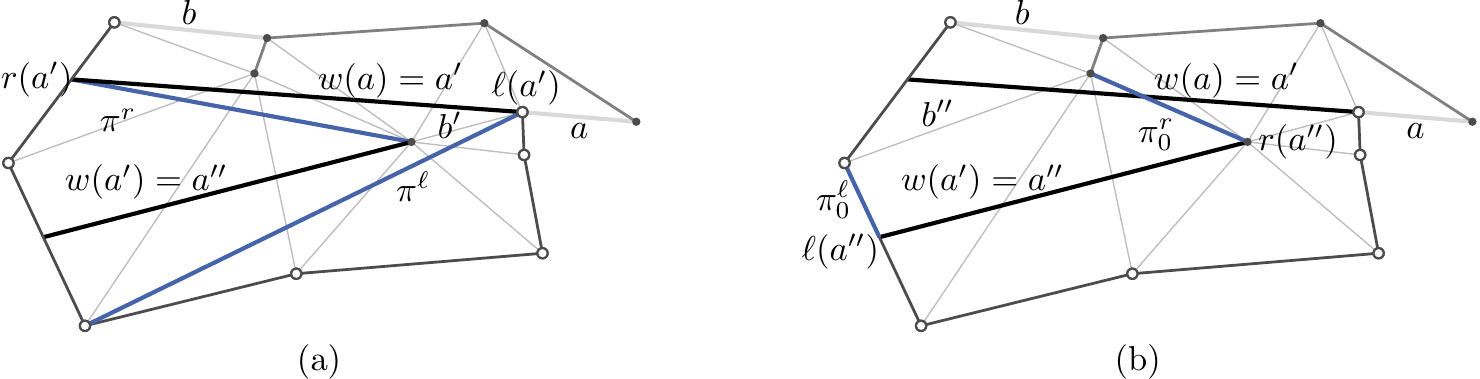}
  \caption{Computing a minimum-link path from $a$ to~$b$ that separates reachable (white) vertices from unreachable (black) vertices. (a)~The window~$w(a')$ computed after the first window $a'$, with (final) shortest paths drawn in blue. (b) The initial shortest paths $\pi_0^\ell$ and $\pi_0^r$ connecting the endpoints of~$a''$ and~$b''$ when computing the next window~$w(a'')$. }
  \label{fig:window-intersecting-path}
\end{figure}

Finally, applying these modifications, we have to clear one last issue to enable correct initialization of the computation of the next window.
Recall that during this initialization, we shortest path becomes the suffix of a previous path; see Section~\ref{sec:min_link_path}.
Figure~\ref{fig:window-intersecting-path} shows an example where this suffix is not available in the modified algorithm.
Assume we want to compute a minimum-link path from $a$ to~$b$. Figure~\ref{fig:window-intersecting-path}a depicts the first window~$w(a) = a'$. Starting from~$a'$, the shortest paths~$\pi^\rig$ and $\pi^\lef$ are computed to obtain the next window~$w(a')$ (note that the right endpoint of~$w(a')$ is an unreachable component consisting of a single vertex).
To compute the next window~$w(a'')$ from~$w(a') = a''$, we first have to compute the initial paths~$\pi_0^\lef$ and $\pi_0^\rig$, as shown in Figure~\ref{fig:window-intersecting-path}b. However, $\pi_0^\rig$ consists of a segment that is not present in the previous right shortest path, because it is visibility-intersecting for this path; see $\pi^\rig$ in Figure~\ref{fig:window-intersecting-path}a.

One way to resolve this problem is to maintain a shortest path~$\pi'$ starting at the corresponding endpoint of the initial edge (\ie, segments are added as in the original algorithm in Section~\ref{sec:min_link_path}). As argued before, this path does not contain self-intersections.
%
%
Then, the initial shortest path is a suffix of~$\pi'$, since the only segments omitted from~$\pi'$ are on the shortest path from the previous window~$w(a) = a'$ to the previous initial edge~$b'$. However, these segments cannot be part of the next initial shortest path, since~$b'$ is fully visible from~$a'$. Hence, the first endpoint of the window~$w(a')$ must be a point in~$\pi'$.

\paragraph{Final Remarks.}

In summary, our algorithm consists of two steps. The first step traverses the reachable boundary in the triangulation of~$B$. The second step runs a modified version of the algorithm described in Section~\ref{sec:min_link_path}, maintaining the next boundary edge index to compute the sequence of important triangles on-the-fly. Both steps are modifications of previous algorithms that clearly maintain their linear running time.
While the resulting polygon~$P$ may intersect itself, it has at most~$\opt + 2$ segments.
To obtain a range polygon $P'$ without self-intersections, we can split $P$ into several non-crossing polygons at intersections. From the resulting smaller polygons, we discard those that contain no vertices of~$\graph_\planar$ or are fully contained in another polygon.
To ensure that $P'$ consists of a single component, according to our primary optimization criterion in Section~\ref{sec:problem_statement}, we can reuse (partial) segments of~$P$ to connect the non-crossing components of~$P'$.
The number of additional segments clearly is linear in the number of self-intersections.

%

\section{Experiments}\label{sec:experiments}

We implemented all approaches in~C++, using~g++ 4.8.3 (-O3) as compiler. Experiments were conducted on a single core of a 4-core Intel Xeon E5-1630v3 clocked at 3.7\,GHz, with 128\,GiB of DDR4-2133 RAM, 10\,MiB of L3, and 256\,KiB of L2 cache.

\paragraph{Input Data.}

%
%
%
%

We evaluate our approaches on a graph representing the road network of Europe, kindly provided by~PTV~AG (\url{ptvgroup.com}).
We extract travel times on road segments from the provided data. This enables us to generate edge weights for computing isochrones.
Energy consumption data for electric vehicles is derived from a realistic energy consumption from a detailed micro-scale emission model~\cite{hrzl-efmph-09}, calibrated to a real Peugeot iOn. It has a battery capacity of 16\,kWh, but we also consider batteries with 85\,kWh, as in high-end Tesla models. Edges for which no reasonable energy consumption can be derived (\eg, due to missing elevation data) are removed from the graph~\cite{Bau13}.
Note that due to recuperation, energy consumption values can become negative for some edges. To ensure that the battery limit is never exceeded, we apply battery constraints~\cite{Bau13} when computing reachable vertices.

The resulting graph is the largest strongly connected component of the remaining subgraph. It has 22\,198\,628 vertices and 51\,088\,095 edges.
To improve spatial locality of the input data, we reorder these vertices according to a vertex partition of the graph~\cite{Bau13}. This slightly improves query performance.
As mentioned in Section~\ref{sec:problem_statement}, we add four bounding box vertices in each corner of the embedding, along with eight edges connecting each vertex to the closest vertex of the input graph and the two closest bounding box vertices.
During planarization, 293\,741 vertices are added and 654\,765 edges are split.
Note that a dummy vertex may intersect more than two original edges, which explains why the number of split edges is more than twice the number of dummy vertices.
They are replaced by 1\,591\,914 dummy edges (creating 6\,294 multi-edges due to overlapping original edges in the given embedding).
After planarization, the resulting graph has 22\,492\,373 vertices and 52\,025\,261 edges.
After triangulating all faces, it has 131\,977\,245 edges in total.

\paragraph{Evaluating Queries.}

\begin{table}[tb!]
\centering
\caption{Results of our algorithms when computing the range of an electric vehicle, for different ranges. We report, for each algorithm and scenario, the number of components of the resulting range polygon~(Cp.), the complexity of the range polygon~(Seg.), the number of self-intersections (Int.), as well as the running time of the algorithm in milliseconds (Time). All figures are average values over 1\,000 random queries. } 
\label{tbl:range}
\setlength{\tabcolsep}{1.5ex}
\begin{tabular}{lrrrrcrrrr}
  \toprule
  & \multicolumn{4}{c}{EV,\,16\,kWh} && \multicolumn{4}{c}{EV,\,85\,kWh} \\
 \cmidrule(lr){2-5}\cmidrule(lr){7-10} Algorithm & Cp. & Seg. & Int. & Time && Cp. & Seg. & Int. & Time \\
 \midrule
\algoExtractReachableComponent & 41 & 19\,396 & --- & 4.50 &  & 131 & 92\,554 & --- & 9.46 \\ 
  \algoTriangularSeparators & 69 & 610 & --- & 4.30 &  & 219 & 1\,973 & --- & 7.78 \\ 
  \algoConnectComponents & 41 & 561 & --- & 10.15 &  & 131 & 1\,820 & --- & 25.11 \\ 
  \algoSelfIntersections & 41 & 549 & 4.79 & 7.52 &  & 131 & 1\,781 & 15.06 & 22.25 \\ 
   \bottomrule
\end{tabular}

\end{table}

We evaluate query scenarios for range visualization of an electric vehicle, as well as isochrones. For electric vehicles, we consider two scenarios, one for medium~(16\,kWh) and one for large~(85\,kWh) battery capacity, corresponding to a range of roughly 100 and 500\,km, respectively.
We compare the results provided by the algorithms proposed in Section~\ref{sec:heuristics}. Each algorithm was tested on the same set of 1\,000 queries from source vertices picked uniformly at random. We denote by~\algoExtractReachableComponent (\emph{range polygon}, extracted \emph{reachable components}) the approach presented in Section~\ref{sec:heuristics:reachable_component_extraction}, by~\algoTriangularSeparators (\emph{triangular separators}) the algorithm from Section~\ref{sec:heuristics:triangle_restricted}, by~\algoConnectComponents (\emph{connecting unreachable} components) the approach from Section~\ref{sec:heuristics:connected_components}, and by~\algoSelfIntersections (\emph{self intersecting} polygons) the algorithm from Section~\ref{sec:heuristics:self_intersecting}.
Below, we only evaluate Steps 2 to 4 of the algorithm outlined in Section~\ref{sec:problem_statement}. The first step, \ie, the computation of the reachable and unreachable parts of the graph, was covered by previous work~\cite{bbdw-fcirn-15}. It is briefly discussed at the end of this section.

\begin{table}[tb!]
\centering
\caption{Overview of the results of our algorithms when computing isochrones for different ranges. Reported results are similar to Table~\ref{tbl:range}. They were obtained by running the same set of 1\,000 random queries} 
\label{tbl:isochrones}
\begin{tabular}{lrrrrcrrrr}
  \toprule
  & \multicolumn{4}{c}{Isochrones,\,60\,min} && \multicolumn{4}{c}{Isochrones,\,500\,min} \\
 \cmidrule(lr){2-5}\cmidrule(lr){7-10} Algorithm & Cp. & Seg. & Int. & Time && Cp. & Seg. & Int. & Time \\
 \midrule
\algoExtractReachableComponent & 53 & 22\,458 & --- & 4.75 &  & 231 & 238\,123 & --- & 20.25 \\ 
  \algoTriangularSeparators & 151 & 1\,076 & --- & 4.65 &  & 694 & 4\,981 & --- & 14.96 \\ 
  \algoConnectComponents & 53 & 913 & --- & 12.11 &  & 231 & 4\,208 & --- & 65.09 \\ 
  \algoSelfIntersections & 53 & 881 & 9.95 & 8.70 &  & 231 & 4\,055 & 45.80 & 51.94 \\ 
   \bottomrule
\end{tabular}

\end{table}

Table~\ref{tbl:range} shows an overview of the results of all heuristics when visualizing the range of an electric vehicle, organized in two blocks. The first considers the medium-range scenario, while the second shows results for large ranges. For each, we report the average number of components, complexity, and the number of self-intersection of the computed range polygons. We also report the average running time in each case.
For~\algoSelfIntersections, the number of components and the complexity are reported as-is after running the modified minimum-link path algorithm described in Section~\ref{sec:heuristics:self_intersecting} (\ie, resulting polygons have the number of self intersections reported in the table). Thus, figures slightly change after resolving the intersections (both the number of components and the complexity may increase).
%
%
We see that all algorithms perform very well in practice, with timings of 25\,ms and below even for large battery capacities. The simpler algorithms,~\algoExtractReachableComponent and~\algoTriangularSeparators are faster by a factor of~2 to~3, compared to the more sophisticated approaches. On the other hand, we see that range polygons generated by~\algoExtractReachableComponent have a much higher complexity, exceeding the optimum by more than an order of magnitude. The heuristic~\algoTriangularSeparators provides much much better results in terms of complexity, but is still outperformed by the other two approaches. Moreover, the triangular separation increases the number of components by almost a factor of $1.7$ (all other approaches in fact compute the minimum number of holes).
Regarding the two more involved approaches,~\algoConnectComponents and~\algoSelfIntersections, we see that the additional effort pays off. Both approaches compute range polygons with the optimal number of components, while keeping the complexity close to the optimum. In fact, we know that each component in the possibly self-intersecting polygon computed by~\algoSelfIntersections requires at most two additional segments compared to an optimal solution. Taking into account that many small components are triangles (which clearly have the optimal complexity), we derive lower bounds on the optimal average complexity of 529 (16\,kWh) and 1\,720 (85\,kWh) for a range polygon with minimum number of components. Hence, our experiments indicate that the average relative error of both~\algoConnectComponents (6\,\%) and~\algoSelfIntersections (4\,\%) is negligible practice. The number of intersections produced by~\algoSelfIntersections is also rather low, but the majority of computed range polygons contains at least one intersection (97.2\,\% for a range of 85\,kWh; not reported in the table).

In Table~\ref{tbl:isochrones}, we provide the according figures for isochrones, considering a harder scenario with isochrones for a range of 500 minutes (roughly eight hours). Resulting isocontours are among the largest on average in our setting (for larger ranges, the border of the network is reached by many queries).
Despite the increase in running time and solution size, we make similar observations as before. All approaches show great performance, with average running times of 65\,ms and below in all cases. Again, the average complexity of range polygons computed by~\algoExtractReachableComponent is larger compared to other heuristics by about a factor of~50, with range polygons consisting of more than 200\,000 segments on average for large ranges.
This clearly justifies the use of our proposed algorithms, since a significant decrease of this number has many advantages when efficient rendering or transmission over mobile networks is an issue. Moreover, a lower number of segments leads to a more appealing visualization for ranges of this order.
For~\algoTriangularSeparators, the number of components now exceeds the optimum by about a factor of~3. Again, the two more sophisticated approaches \algoConnectComponents and~\algoSelfIntersections yield best results with average relative errors bounded by 7\,\% and 3\,\%, respectively. 

\begin{table}[tb!]
\centering
\caption{Running times of different phases of the algorithms (where applicable) for isochrones with a range of 500 minutes. For each algorithm, we report the total running time (Total) together with the running time for transferring the the input to the planar graph (TP), extracting the border regions (BE), connecting components (CC), the range polygon computation with minimum-link paths (RP), and the test for self-intersections~(SI). All timings are given in milliseconds.} 
\label{tbl:phases}
\begin{tabular}{lrrrrrr}
  \toprule
  Algorithm & TP & BE & CC & RP & SI & Total \\
 \midrule
\algoExtractReachableComponent & 8.21 & 12.01 & --- & --- & --- & 20.25 \\ 
  \algoTriangularSeparators & 8.22 & --- & --- & 6.45 & --- & 14.96 \\ 
  \algoConnectComponents & 8.23 & 26.66 & 22.99 & 7.81 & --- & 65.09 \\ 
  \algoSelfIntersections & 8.20 & 31.79 & --- & 9.53 & 2.34 & 51.94 \\ 
   \bottomrule
\end{tabular}

\end{table}

For the hardest scenario considered so far (isochrones, 500\,min.), we report more detailed information on running times of the different phases of all algorithms in Table~\ref{tbl:phases}. Note that the total running time slightly differs from the sum of all subphases, since it was determined independently.
The planarization phase~(TP) consists of the linear sweeps described in Section~\ref{sec:range_query}. Since the same work needs to be done for all approaches, the running time is identical in all cases (except for noise in the measurement). Of course, the relative amount spent in this step differs per algorithm. In case of~\algoTriangularSeparators it requires more than half of the total running time.
The time to extract the border regions~(BE) applies to all algorithms except~\algoTriangularSeparators, where this is done implicitly by checking the reachability of vertices. Since~\algoExtractReachableComponent extracts only the reachable boundary, this phase takes less than half the time compared to~\algoConnectComponents (the unreachable boundary is typically larger). Finally,~\algoSelfIntersections spends most time in this step, since it runs the extraction on the triangulated graph, which is significantly more dense (recall that~\algoSelfIntersections in fact only extracts the reachable border, similar to~\algoExtractReachableComponent, so this phase is slower by more than a factor of~$2.5$).
Connecting unreachable components (CC) is only necessary in~\algoConnectComponents and takes less time than the extraction of border regions.
Computing the actual range polygon takes a similar amount of time for all approaches that run this phase (6 to 10\,ms). For~\algoTriangularSeparators, it is slightly faster, since the algorithm works only on important triangles, reducing the number of visited triangles and simplifying the algorithm (see Section~\ref{sec:heuristics:triangle_restricted}). On the other hand,~\algoSelfIntersections is the slowest approach in this phase. This can be explained by the overhead caused by the modifications described in Section~\ref{sec:heuristics:self_intersecting}. Moreover, in contrast to \algoTriangularSeparators and~\algoConnectComponents, there are no artificial edges in the border regions. Hence, windows computed by~\algoSelfIntersections are longer on average, increasing the number of triangles visited by the algorithm.

In summary, we see that extracting the border region takes a major fraction of the total effort for all approaches that compute~$B$ explicitly. Despite the algorithmic simplicity of this phase, this can be explained by the size of the border region. In our experiments, it consisted of more than 500\,000 segments on average. On the other hand, only a fraction of the triangles in the border regions are visited by the minimum-link path algorithm.

\paragraph{Evaluating Scalability.}

\begin{figure}[tb!]
 \centering
 \input{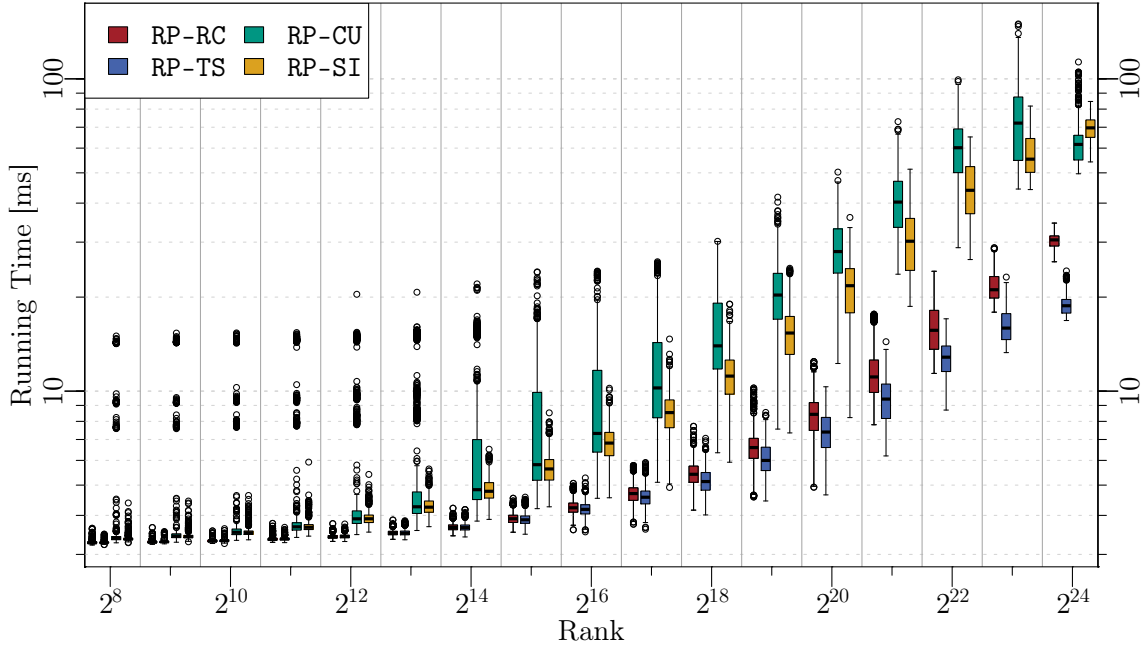}%
 \caption{Running times of all approaches subject to Dijkstra rank. Smaller ranks indicate queries of shorter range. For each rank, we report results of 1\,000 random queries.}
 \label{fig:tt-rank-boxplot}
\end{figure}

Figure~\ref{fig:tt-rank-boxplot} analyzes the scalability of our algorithms, following the methodology of Dijkstra ranks~\cite{Bas14}. The Dijkstra rank of a shortest-path query is the the number of queue extractions performed by Dijkstra's algorithm (assuming that the algorithm stops once the target was found). Thus, higher ranks reflect harder queries. To obtain queries of different ranks, we ran 1\,000 queries of the modified version of Dijkstra's algorithm from Section~\ref{sec:range_query}, with infinite range and from sources chosen uniformly at random. For a query from some source~$s$, consider the resource consumption~$c$ of the corresponding vertex that was extracted from the queue in step~$2^i$ of the algorithm (the maximum rank is bounded by the graph size). We consider a query from~$s$ with range~$c$ as a query of rank~$2^i$. For each rank in~$\{2^1, \dots, 2^{\log |V|}\}$, we evaluate the 1\,000 queries generated this way.

We see that query times of all approaches increase with the Dijkstra rank, which correlates well with the complexity of the border region. Moreover, scaling behavior is similar for all approaches. In accordance to our theoretical findings, it increases linearly in the size of the border region for queries beyond a rank of $2^{12}$. For queries of lower rank, transferring the input to the planar graph dominates running time (which is linear in the graph size). The approach~\algoTriangularSeparators is consistently the fastest approach on average for all ranks beyond $2^{16}$. However, except for a few outliers, all queries run in well below 100\,ms. For more local queries (\ie, smaller ranges), query times are much faster (20\,ms and below if the rank is at most $2^{20}$, corresponding to about a million vertices visited by Dijkstra's algorithm).
Interestingly, the more expensive approaches have a higher variance and produce more outliers, which is explained by their more complex phases. For example, the performance of the BFS used by~\algoConnectComponents heavily depends on how close unreachable components of the border region are in the dual graph.

\paragraph{Evaluating the Computation of Minimum-Link Paths.}

\begin{table}[tb!]
\centering
\caption{Running times of the minimum-link path algorithm. For each scenario, we report the number of segments in the input polygon ($|P|$), the number of triangles visited by the algorithm (v.\,Tr.), the number of links of the computed path (Seg.), and the running time in milliseconds, each averaged over 1\,000 queries.} 
\label{tbl:mlp}
\setlength{\tabcolsep}{1.5ex}
\begin{tabular}{lrrrr}
  \toprule
  Scenario & $|P|$ & v.\,Tr. & Seg. & Time \\
 \midrule
EV, 16 kWh & 134\,049 & 9\,334 & 416 & 0.72 \\ 
  Iso, 60 min & 135\,112 & 11\,965 & 701 & 1.03 \\ 
  EV, 85 kWh & 357\,335 & 33\,030 & 1\,329 & 3.14 \\ 
  Iso, 500 min & 637\,224 & 69\,398 & 3\,204 & 6.57 \\ 
   \bottomrule
\end{tabular}

\end{table}

We take a closer look at the performance of our algorithm for computing minimum-link paths introduced in Section~\ref{sec:min_link_path}. To properly evaluate the algorithm in the context of our experimental setting, we proceed as follows. For each query, we consider the largest corresponding border region (\wrt number of segments). To obtain a polygon without holes (as required by the algorithm), we first run our heuristic to connect all unreachable components described in Section~\ref{sec:heuristics:connected_components}. Then, we add an arbitrary boundary edge to the modified border region, and compute a minimum-link path that connects both sides of this boundary edge. Results are shown in Table~\ref{tbl:mlp}. Each scenario is based on the respective sets of random queries used in Tables~\ref{tbl:range} and~\ref{tbl:isochrones}.

Clearly, each scenario represents a certain level of difficulty, with the average complexity of the input polygon ranging from some 100\,000 to 600\,000 segments. Clearly, the number of visited triangles, the number of segments of the resulting path, and the running time increase with the complexity of the input. However, we also see in Table~\ref{tbl:mlp} that the algorithm performs excellently in practice, computing minimum-link paths in less than 7 milliseconds, even for input polygons with more than half a million vertices.
Somewhat surprisingly, the isochrone scenarios (60\,min) is slightly harder then the range scenario (16\,kWh), despite a similar input complexity.
This can be explained by the different shape of the respective border regions.
Isochrones in road networks reach further on highways and other fast roads, leading to spike-like shapes in the resulting border regions. On the other hand, isocontours representing the range of an electric vehicle typically have a more circular shape (highways allow to move faster, but consume more energy). Consequently, range polygons for isochrones require more segments and yield the more challenging scenario.

\paragraph{Computation of the Reachable Subgraph.}

In all scenarios considered above, we ignored the computational effort to obtain the reachable subgraph (Step~1 of the general approach in Section~\ref{sec:problem_statement}). We proposed a variant of Dijkstra's algorithm to achieve this in Section~\ref{sec:range_query}. However, its performance would be the bottleneck of the overall running time in all scenarios. For large ranges, it takes several seconds on average. Similar observations for variants of Dijkstra's algorithm running on large-scale road networks were made in the context of speedup techniques for shortest-path algorithms~\cite{Bas14}. Recently, it was shown that algorithms producing output that is similar to Step~1 can be made practical using preprocessing-based techniques~\cite{bbdw-fcirn-15}. As a result, running times below 50 milliseconds can be achieved for this step (and even less when parallelized).
Since these techniques can easily be adapted to produce the output of Step~1 required by our algorithm, our approaches enable the visualization of isocontours in road networks in less than 100 milliseconds in total. Thereby, we enable interactive applications even on road networks of continental scale.

\section{Conclusion}\label{sec:conclusion}

In this work, we proposed several approaches for computing isocontours in large-scale road networks.
Given the subgraph that is reachable within the given resource limit, this problem boils down to computing a geometric representation of this subgraph.
We introduced range polygons, following three major objectives, namely, yielding exact results, low complexity, and practical performance on large inputs.
To this end, we adapted previous approaches to our scenario and presented three novel algorithms, all of which use as key ingredient a new linear-time algorithm for minimum-link paths in simple polygons.
Our experimental evaluation clearly showed that our approaches are suitable for interactive applications on inputs of continental scale, providing solutions of almost optimal complexity.

Regarding future work, there are several interesting open issues and room for further improvements.
First, our techniques exploit the fact that the reachable boundary of a border region is always connected, \ie, $|R| = 1$. This might not be the case in related scenarios, such as multi-source isochrones or in multimodal networks, where one can disembark public transit vehicles only at certain points~\cite{Gam11}. Thus, it would be interesting to know whether our approaches can be extended to the case of~$|R| > 1$.
Moreover, Gamper et al.~\cite{Gam11} consider the extent to which some reachable edge can be passed in their definition of isochrones. This is relevant particularly for short ranges, or if the graph consists of edges with very long distance. Hence, we could modify our approaches to handle this slightly different definition.
From an aesthetic point of view, one could aim at extending our approaches to avoid very long straight segments in the range polygon. For example, an ocean at the boundary of the reachable area may correspond to a huge face in the planar representation of the graph. The resulting range polygon may contain segments reaching far into this face, deteriorating the visualization. This could be prevented, \eg, by adding dummy faces and edges to the graph. This may even decrease the running time the heuristic in Section~\ref{sec:heuristics:connected_components}, if the unreachable boundary becomes smaller (though, it lead only to a mild speedup in preliminary experiments).
On the other hand, one could also aim at further line simplification, at the cost of inexact results. However, such methods should avoid intersections between different components of the range polygon (\ie, maintain its topology), and error measures should consider the graph-based distance from the source to parts of the network that are incorrectly classified by the range polygon (since close vertices \wrt Euclidean distance may in fact have a long distance in the graph).
Finally, another interesting open problem is the consideration of continuous range visualization for a moving vehicle. Instead of computing the isocontours from scratch, one could try to reuse previously computed information.

\paragraph{Acknowledgements.}
We thank Roman Prutkin for interesting discussions.

\bibliographystyle{plain}
\bibliography{references}

\end{document}